%% file: ms.tex
\documentclass[sigconf]{acmart}

\usepackage{booktabs} % For formal tables
\usepackage{tabularx}
\usepackage{xspace}
\usepackage{enumitem}
\usepackage{xcolor,colortbl}
\usepackage{amsmath}
\usepackage{multirow}
\usepackage{hyperref}
\usepackage{graphicx}
\usepackage[center]{subfigure}
\usepackage[ruled,linesnumbered]{algorithm2e} % algorithm
\usepackage[]{algorithmic} % algorithm\textbf{}
\usepackage{tablefootnote}
\usepackage{threeparttable}
\usepackage{makecell}
\usepackage[font={small}]{caption}
\usepackage{diagbox}
\usepackage{empheq}

\newcommand*\widefbox[1]{\fbox{\hspace{2em}#1\hspace{2em}}}
\begin{document}
	
\copyrightyear{2019} 
\acmYear{2019} 
\setcopyright{acmcopyright}
\acmConference[KDD '19]{The 25th ACM SIGKDD Conference on Knowledge Discovery and Data Mining}{August 4--8, 2019}{Anchorage, AK, USA}
\acmBooktitle{The 25th ACM SIGKDD Conference on Knowledge Discovery and Data Mining (KDD '19), August 4--8, 2019, Anchorage, AK, USA}
\acmPrice{15.00}
\acmDOI{10.1145/3292500.3330946}
\acmISBN{978-1-4503-6201-6/19/08}

\settopmatter{printacmref=true}
\fancyhead{}

\title{Fast and Accurate Anomaly Detection in Dynamic Graphs \\ with a Two-Pronged Approach}

\author{
	Minji Yoon\footnotemark[1],~
	Bryan Hooi\footnotemark[2],~
	Kijung Shin\footnotemark[3],~ 
	Christos Faloutsos\footnotemark[1]
}
\affiliation{
	\institution{\footnotemark[1]~School of Computer Science, Carnegie Mellon University, Pittsburgh, PA, USA} 
	\institution{\footnotemark[2]~School of Computing, National University of Singapore, Singapore}
	\institution{\footnotemark[3]~School of Electrical Engineering, KAIST, South Korea}
	\institution{minjiy@cs.cmu.edu, bhooi@andrew.cmu.edu, kijungs@kaist.ac.kr, christos@cs.cmu.edu}
}

\input{dfn}

\begin{abstract}
\input{000abstract.tex}

\end{abstract}

%\begin{CCSXML}
%	<ccs2012>
%	<concept>
%	<concept_id>10002951.10003227.10003351</concept_id>
%	<concept_desc>Information systems~Data mining</concept_desc>
%	<concept_significance>300</concept_significance>
%	</concept>
%	<concept>
%	<concept_id>10003033.10003106.10003114.10011730</concept_id>
%	<concept_desc>Networks~Online social networks</concept_desc>
%	<concept_significance>500</concept_significance>
%	</concept>
%	</ccs2012>
%\end{CCSXML}

%\ccsdesc[300]{Information systems~Data mining}
%\ccsdesc[500]{Networks~Online social networks}

%\keywords{Anomaly Detection, Online Algorithms, PageRank}

\maketitle

\section{Introduction}
\label{sec:introduction}
\input{010introduction.tex}

\section{Related Work}
\label{sec:related_works}
\input{050related_works.tex}

\section{Preliminaries}
\label{sec:preliminaries}
\input{020preliminary.tex}

\section{Proposed Method}
\label{sec:proposed_method}
\input{030proposed_method.tex}

\section{Experiments}
\label{sec:experiments}
\input{040experiments.tex}

\section{Conclusion}
\label{sec:conclusion}
\input{060conclusion.tex}

\section*{Acknowledgment}
\label{sec:ack}
\input{080acknowledgements.tex}

\bibliographystyle{ACM-Reference-Format}
\bibliography{BIB/myref,BIB/additions}

\newpage
\appendix
\section{Supplement}
\label{sec:appendix}
\input{070appendix.tex}

\end{document}

%% file: dfn.tex
\newtheorem{lemm}{Lemma}
\newtheorem{defi}{Definition}
\newtheorem{theo}{Theorem}
\newtheorem{assumption}{Assumption}
\newtheorem{observation}{Observation}
\newtheorem{problem}{Problem}
\newtheorem{algo}{Algorithm}
\newtheorem{property}{Property}
\renewcommand{\qed}{$\hfill \blacksquare$}

\newcommand{\reminder}[1]{{{\textcolor{red}{\bf [#1]}}}}
\newcommand{\hide}[1]{}

\newcommand{\method}{\textsc{AnomRank}\xspace}

\newcommand{\prS}{\textsc{ScoreS}\xspace}
\newcommand{\prE}{\textsc{ScoreW}\xspace}
\newcommand{\prSE}{\textsc{ScoreS/W}\xspace}

\newcommand{\anS}{\textsc{AnomRankS}\xspace}
\newcommand{\anE}{\textsc{AnomRankW}\xspace}
\newcommand{\anEF}{\textsc{AnomRankW-1st}\xspace}
\newcommand{\anES}{\textsc{AnomRankW-2nd}\xspace}

\newcommand{\structure}{\textsc{AnomalyS}\xspace}
\newcommand{\edge}{\textsc{AnomalyW}\xspace}

\newcommand{\sceneS}{\textsc{InjectionS}\xspace}
\newcommand{\sceneE}{\textsc{InjectionW}\xspace}

\newcommand{\mat}[1]{\mathbf{#1}}
\newcommand{\set}[1]{\mathbf{#1}}
\newcommand{\vect}[1]{\mathbf{#1}}

\newcommand{\p}{\vect{p}} % for a pagerank vector
\newcommand{\ps}{\vect{p}_s} % for a pagerank vector
\newcommand{\pe}{\vect{p}_w} % for a pagerank vector

\renewcommand{\a}{\vect{a}}
\newcommand{\as}{\vect{a}_s}
\renewcommand{\ae}{\vect{a}_w}

\newcommand{\dt}{\Delta t}
\newcommand{\dm}{\frac{\Delta m}{\dt}}
\newcommand{\dmm}{\frac{\Delta^2 m}{\dt^2}}
\newcommand{\psone}{\ps'}
\newcommand{\peone}{\pe'}
\newcommand{\pstwo}{\ps''}
\newcommand{\petwo}{\pe''}

\newcommand{\pso}{\vect{p}_s^o} % for a pagerank vector
\newcommand{\psn}{\vect{p}_s^n} % for a pagerank vector
\newcommand{\peo}{\vect{p}_w^o} % for a pagerank vector
\newcommand{\pen}{\vect{p}_w^n} % for a pagerank vector

\renewcommand{\b}{\vect{b}} % for a starting vector
\newcommand{\bs}{\vect{b}_s} % for a starting vector
\newcommand{\be}{\vect{b}_w} % for a starting vector
\newcommand{\bd}{\Delta\vect{b}} % for a starting vector
\newcommand{\bed}{\Delta\vect{b}_w} % for a starting vector
\newcommand{\bedn}{\Delta\vect{b}_{w_n}} % for a starting vector
\newcommand{\bedo}{\Delta\vect{b}_{w_o}} % for a starting vector

\newcommand{\A}{\mat{A}} % adjacency matrix
\newcommand{\NA}{\mat{\tilde{A}}} % normalized adjacency matrix
\newcommand{\NAT}{\mat{\tilde{A}}^{\top}} % normalized
\newcommand{\B}{\mat{B}} % adjacency matrix
\newcommand{\NB}{\mat{\tilde{B}}} % normalized adjacency matrix
\newcommand{\NBT}{\mat{\tilde{B}}^{\top}} % normalized
\newcommand{\DA}{\Delta\mat{A}} 

\newcommand{\AS}{\mat{A}_s} % adjacency matrix
\newcommand{\NAS}{\mat{\tilde{A}}_s} % normalized adjacency matrix
\newcommand{\NATS}{\mat{\tilde{A}}_s^{\top}} % normalized
\newcommand{\BS}{\mat{B}_s} % adjacency matrix
\newcommand{\NBS}{\mat{\tilde{B}}_s} % normalized adjacency matrix
\newcommand{\NBTS}{\mat{\tilde{B}}_s^{\top}} % normalized
\newcommand{\DAS}{\Delta\mat{A}_s} 

\renewcommand{\AE}{\mat{A}_w} % adjacency matrix
\newcommand{\NAE}{\mat{\tilde{A}}_w} % normalized adjacency matrix
\newcommand{\NATE}{\mat{\tilde{A}}_w^{\top}} % normalized
\newcommand{\BE}{\mat{B}_w} % adjacency matrix
\newcommand{\NBE}{\mat{\tilde{B}}_w} % normalized adjacency matrix
\newcommand{\NBTE}{\mat{\tilde{B}}_w^{\top}} % normalized
\newcommand{\DAE}{\Delta\mat{A}_w} 

\newcommand{\DASN}{\Delta\mat{A}_{s_n}} 
\newcommand{\DASO}{\Delta\mat{A}_{s_o}}
\newcommand{\DAEN}{\Delta\mat{A}_{w_n}} 
\newcommand{\DAEO}{\Delta\mat{A}_{w_o}}

\newcommand{\note}[1]{\textcolor{red}{#1}}
\renewcommand{\comment}[1]{\textcolor{blue}{#1}}

\newcommand{\oddball}{\textsc{OddBall}\xspace}
\newcommand{\scan}{\textsc{SCAN}\xspace}
\newcommand{\autopart}{\textsc{Autopart}\xspace}
\newcommand{\nnrmf}{\textsc{NNRMF}\xspace}
\newcommand{\fraudar}{\textsc{Fraudar}\xspace}
\newcommand{\coreA}{\textsc{Core-A}\xspace}
\newcommand{\catchsync}{\textsc{CatchSync}\xspace}
\newcommand{\copycatch}{\textsc{CopyCatch}\xspace}
\newcommand{\crossspot}{\textsc{CrossSpot}\xspace}
\newcommand{\mzoom}{\textsc{M-Zoom}\xspace}
\newcommand{\densealert}{\textsc{DenseAlert}\xspace}
\newcommand{\sedanspot}{\textsc{SedanSpot}\xspace}
\newcommand{\spotlight}{\textsc{Spotlight}\xspace}

\newcommand{\codeurl}{\footnote{\url{https://github.com/minjiyoon/anomrank}}}

%% file: 000abstract.tex
Given a dynamic graph stream, how can we detect the sudden appearance of anomalous patterns, such as link spam, follower boosting, or denial of service attacks? Additionally, can we categorize the types of anomalies that occur in practice, and theoretically analyze the anomalous signs arising from each type?

In this work, we propose \method, an online algorithm for anomaly detection in dynamic graphs.
\method uses a two-pronged approach defining two novel metrics for anomalousness.
Each metric tracks the derivatives of its own version of a `node score' (or node importance) function.
This allows us to detect sudden changes in the importance of any node.
We show theoretically and experimentally that the two-pronged approach successfully detects two common types of anomalies: sudden weight changes along an edge, and sudden structural changes to the graph. 
\method is {\bf (a) Fast and Accurate:} up to {\it 49.5$\times$ faster} or {\it 35\% more accurate} than state-of-the-art methods,
{\bf (b) Scalable:} linear in the number of edges in the input graph, processing
millions of edges within 2 seconds on a stock laptop/desktop, and
{\bf (c) Theoretically Sound:} providing theoretical guarantees of the two-pronged approach.

%% file: 010introduction.tex
Network-based systems including computer networks and social network services have been a focus of various attacks.
In computer networks, distributed denial of service (DDOS) attacks use a number of machines to make connections to a target machine to block their availability.
In social networks, users pay spammers to "Like" or "Follow" their page to manipulate their public trust.
By abstracting those networks to a graph, we can detect those attacks by finding suddenly emerging anomalous signs in the graph. 

Various approaches have been proposed to detect anomalies in graphs, the majority of which focus on static graphs~\cite{akoglu2010oddball,chakrabarti2004autopart,hooi2017graph,jiang2016catching,kleinberg1999authoritative,shin2018patterns,tong2011non}.
However, many real-world graphs are dynamic, with timestamps indicating the time when each edge was inserted/deleted.
Static anomaly detection methods, which are solely based on static connections, miss useful temporal signals of anomalies.

Several approaches~\cite{eswaran2018spotlight,shin2017densealert,eswaran2018sedanspot} have been proposed to detect anomalies on dynamic graphs (we review these in greater detail in Section \ref{sec:related_works} and Table \ref{tab:salesman}).
However, they are not satisfactory in terms of accuracy and speed.
Accurately detecting anomalies in near real-time is important in order to cut back the impact of malicious activities and start recovery processes in a timely manner.

\begin{figure*}[!t]
	\centering
	\vspace{-5mm}
	\subfigure[Speed and precision of \method]
	{
		\label{fig:perf:running_time}
		\includegraphics[width=.3\linewidth]{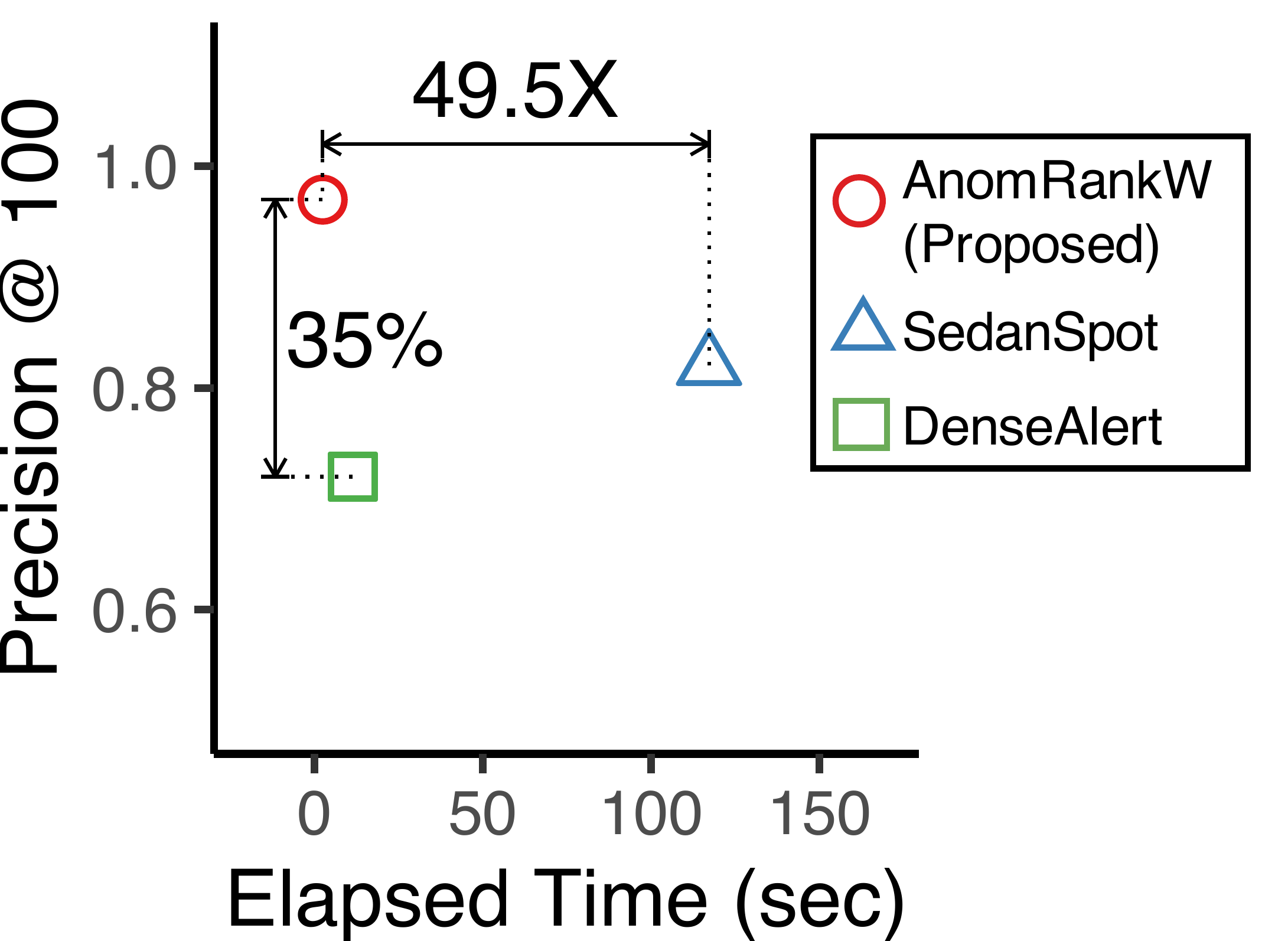}	
	}
	\subfigure[Scalability of \method]
	{
		\label{fig:perf:scalability}
		\includegraphics[width=.21\linewidth]{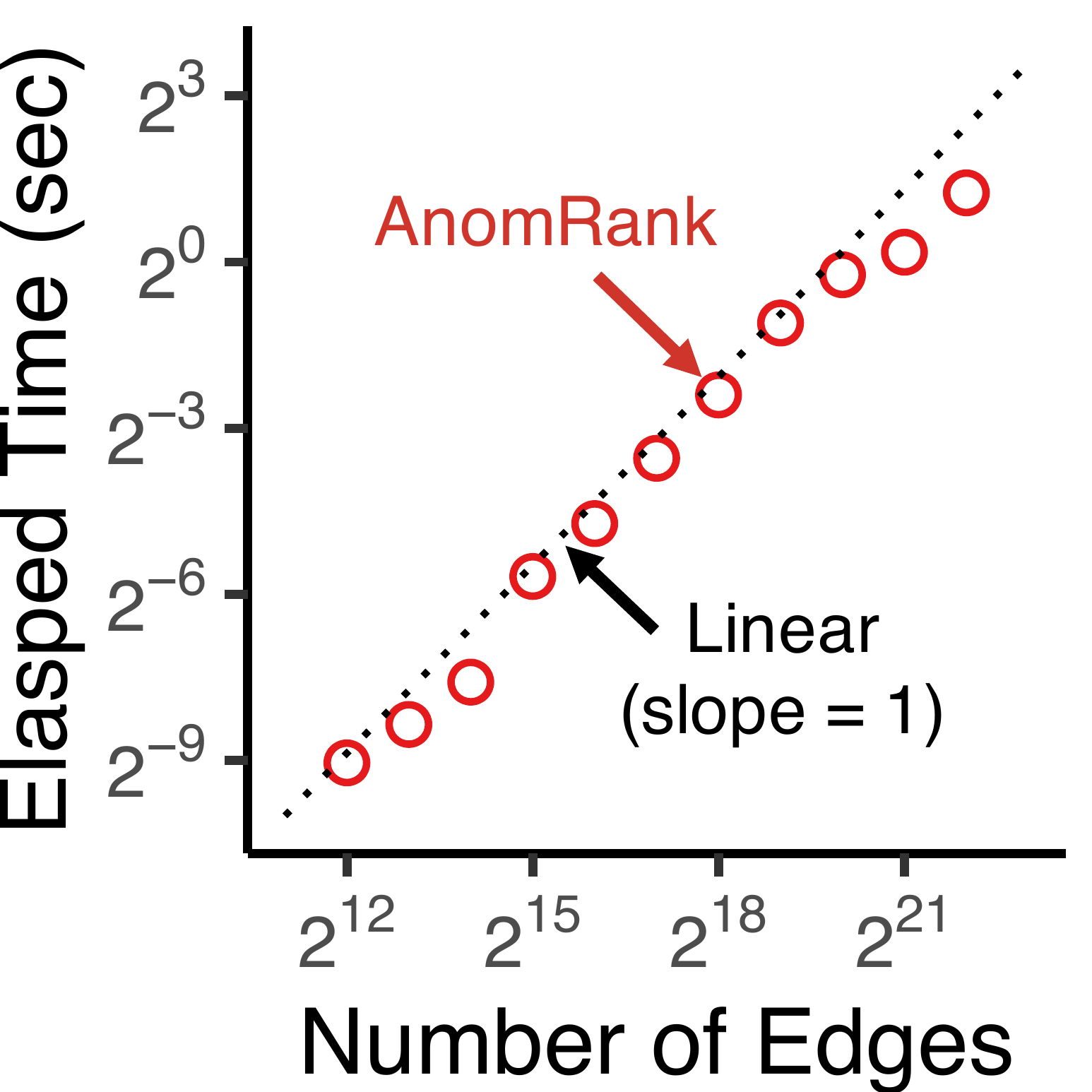}	
	}
	\subfigure[Effectiveness of \method]
	{
		\label{fig:perf:anomaly_score}
		\includegraphics[width=.4\linewidth]{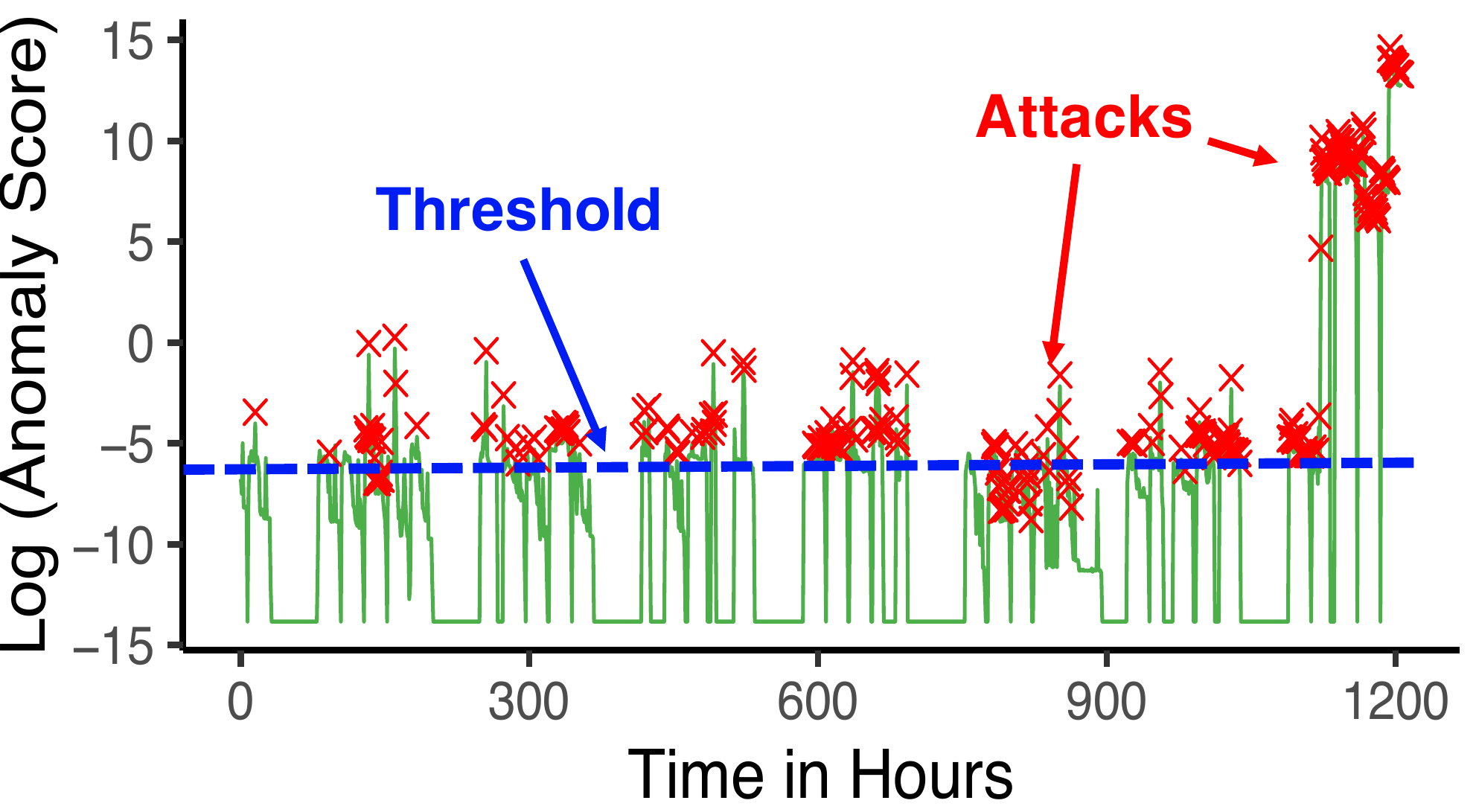}	
	}
	\caption
	{ \underline{\smash{\method is accurate, fast, and scalable}}:
		(a) Precision at top-100 and running time on the DARPA dataset:		
		\method is significantly faster than state-of-the-art methods while achieving higher accuracy.
		(b)	\method scales linearly with the number of edges in the input dynamic graph.
		(c) \method computes anomaly scores (green line) on the DARPA dataset. 
		Red crosses indicate ground truth anomalies; 
		the blue line is an anomalousness threshold.
		$77$\% of green spikes above the blue line are true positives;
		see Section~\ref{sec:experiments} for details.
	}
	\label{fig:perf}
\end{figure*}

In this paper, we propose \method, a fast and accurate online algorithm for detecting anomalies in dynamic graphs with a two-pronged approach.
We classify anomalies in dynamic graphs into two types: \structure and \edge.
\structure denotes suspicious changes to the \emph{structure} of the graph, such as through the addition of edges between previously unrelated nodes in spam attacks.
\edge indicates anomalous changes in the \emph{weight} (i.e. number of edges) between connected nodes, such as suspiciously frequent connections in port scan attacks.

Various node score functions have been proposed to map each node of a graph to an importance score:
PageRank~\cite{page1999pagerank}, HITS and its derivatives (SALSA)~\cite{kleinberg1999authoritative, DBLP:journals/tois/LempelM01}, Fiedler vector~\cite{chungspectral}, etc.
Our intuition is that anomalies induce sudden changes in node scores. %PageRank scores.
Based on this intuition, \method focuses on the 1st and 2nd order derivatives of node scores to detect anomalies with large changes in node scores within a short time.
To detect \structure and \edge effectively, we design two versions of node score functions based on characteristics of these two types of anomalies.
Then, we define two novel metrics, \anS and \anE, which measure the 1st and 2nd order derivatives of our two versions of node score functions, as a barometer of anomalousness.
We theoretically analyze the effectiveness of each metric on the corresponding anomaly type, and provide rigid guarantees on the accuracy of \method. 
Through extensive experiments with real-world and synthetic graphs, we demonstrate the superior performance of \method over existing methods.
%Table~\ref{tab:salesman} and Figure~\ref{fig:perf} show a comparison of \method and existing methods. 
Our main contributions are:
\begin{itemize}[leftmargin=10pt]
	\item 
	{\textbf{Online, two-pronged approach:}
		We introduce \method, an online detection method, for two types of anomalies
		(\structure, \edge) in dynamic graphs.
		% \method defines two novel metrics tracking 1st and 2nd order derivatives of two versions of PageRank and detects anomalies in near real-time.
	}
	\item {
	\textbf{Theoretical guarantees:}
		We prove the effectiveness of our proposed metrics, \anS and \anE, theoretically (Theorems \ref{theorem:up_d1d2_s} and \ref{theorem:up_d1d2_e}).
		%Theorems \ref{theorem:up_d1d2_s} and \ref{theorem:up_d1d2_e} provide lower bounds on the accuracy of our proposed metrics, \anS and \anE. 
	}
	\item {
	\textbf{Practicality:}
		Experiments on public benchmarks show that \method outperforms state-of-the-art competitors, being up to \textit{$49.5 \times$} faster or \textit{$35\%$} more accurate (Figure \ref{fig:perf}).
		Moreover, thanks to its two-pronged approach, it spots anomalies that the competitors miss (Figure \ref{fig:enron}).
	}
\end{itemize}

{\bf Reproducibility}: our code and data are publicly available\codeurl.
The paper is organized in the usual way (related work, preliminaries, proposed method, experiments, and conclusions).

%% file: 050related_works.tex
We discuss previous work on detecting anomalous entities (nodes, edges, events, etc.) on static and dynamic graphs.
See \cite{akoglu2015graph} for an extensive survey on graph-based anomaly detection.

{\bf Anomaly detection in static graphs} can be described under the following categories:
\begin{itemize}[leftmargin=10pt]
	\item {\it Anomalous Node Detection}: \cite{akoglu2010oddball} extracts egonet-based features and finds empirical patterns with respect to the features. Then, it identifies nodes whose egonets deviate from the patterns. %, including the count of triangles, total weight, and principal eigenvalue,
	%\cite{jiang2016catching} computes node features, including degree and authoritativeness \cite{kleinberg1999authoritative}, and spots nodes whose neighbors are notably close in the feature space.
	\cite{xu2007scan} groups nodes that share many neighbors and spots nodes that cannot be assigned to any community.
	\item {\it Anomalous Edge Detection}: \cite{chakrabarti2004autopart} encodes the input graph based on similar connectivity between nodes, then spots edges whose removal significantly reduces the total encoding cost.
	\cite{tong2011non} factorizes the adjacency matrix and flags edges which introduce high reconstruction error as outliers. 
	%compares the input adjacency matrix and its approximation obtained by matrix factorization. Then, edges corresponding to entries with high reconstruction error are flagged as outliers.
	\item {\it Anomalous Subgraph Detection}: \cite{hooi2017graph} and \cite{shin2018patterns} measure the anomalousness of nodes and edges, then find a dense subgraph consisting of many anomalous nodes and edges. 
\end{itemize}

\noindent{\bf Anomaly detection in dynamic graphs} can also be described under the following categories:
\begin{itemize}[leftmargin=10pt]
	\item {\it Anomalous Node Detection:} 
	\cite{sun2006beyond} approximates the adjacency matrix of the current snapshot based on incremental matrix factorization.
	Then, it spots nodes corresponding to rows with high reconstruction error.
	%Given a sequence of graph snapshots,
	\cite{wang2015localizing} computes nodes features (degree, closeness centrality, etc) in each graph snapshot.	
	Then, it identifies nodes whose features are notably different from their previous values and the features of nodes in the same community.
	%Given an edge stream, \cite{yu2013anomalous} detects nodes whose egonets suddenly and significantly change.
	\item {\it Anomalous Edge Detection:} 
	%Given an edge stream,
	\cite{eswaran2018sedanspot} detects edges that connect sparsely-connected parts of a graph.
	\cite{ranshous2016scalable} spots edge anomalies based on their occurrence, preferential attachment and mutual neighbors.
	\item {\it Anomalous Subgraph Detection}: 
	%Given a graph with timestamps on edges, 
	\cite{beutel2013copycatch} spots near-bipartite cores where each node is connected to others in the same core densly within a short time. 
	%Given a sequence of graph (or tensor) snapshots,
	\cite{jiang2016catching} and \cite{shin2018fast} detect groups of nodes who form dense subgraphs in a temporally synchronized manner.
	%Given an edge stream,
	\cite{shin2017densealert} identifies dense subtensors created within a short time.
	\item {\it Event Detection}:	
	%Given a sequence of graph snapshots, 
	\cite{eswaran2018spotlight,aggarwal2011outlier,koutra2016deltacon, henderson2010metric} detect the following events: sudden appearance of many unexpected edges \cite{aggarwal2011outlier},
	sudden appearance of a dense graph \cite{eswaran2018spotlight}, sudden drop in the similarity between two consecutive snapshots \cite{koutra2016deltacon}, and sudden prolonged spikes and lightweight stars~\cite{henderson2010metric}. 
	%and sudden change in the community structure~\cite{sun2007graphscope}. 
\end{itemize}

Our proposed \method is an anomalous event detection method with fast speed and high accuracy. 
It can be easily extended to localize culprits of anomalies into nodes and substructures (Section~\ref{sec:algorithm}), and it detects various types of anomalies in dynamic graphs in a real-time.
Table \ref{tab:salesman} compares \method to existing methods. 

\begin{table}[!htb]
	\centering
	\caption{
		\underline{\smash{\method out-features competitors}}: comparison of our proposed \method and existing methods for anomaly detection in dynamic graphs.
	}
	\label{tab:salesman}
	\begin{tabular}{l|llllll|l}
		\hline
		\diagbox[width=3cm, height=1.8cm]{Property}{Method}&\rotatebox[origin=c]{90}{Oddball~\cite{akoglu2010oddball}} &
		\rotatebox[origin=c]{90}{MetricFor.~\cite{henderson2010metric}} & 
		\rotatebox[origin=c]{90}{CC, CS\cite{beutel2013copycatch,jiang2016catching}} &
		\rotatebox[origin=c]{90}{DenseAlert~\cite{shin2017densealert}} &
		\rotatebox[origin=c]{90}{SpotLight~\cite{eswaran2018spotlight}} &
		\rotatebox[origin=c]{90}{SedanSpot~\cite{eswaran2018sedanspot}} & 
		\rotatebox[origin=c]{90}{~\textbf{\method}~} \\ \hline
		Real-time detection$^*$ & & & & \checkmark & & & \checkmark\\
		Allow edge deletions & & & & \checkmark & \checkmark & & \checkmark\\
		Structural anomalies & \checkmark &\checkmark & & & & & \checkmark\\
		Edge weight anomalies & \checkmark & \checkmark & \checkmark & \checkmark & \checkmark & \checkmark & \checkmark\\ \hline
	\end{tabular}
\begin{tablenotes}
	\small
	\item $^*$compute $1$M edges within $5$ seconds.
\end{tablenotes}
\vspace{-1mm}
\end{table}

%% file: 020preliminary.tex
Table~\ref{tab:symbols} gives a list of symbols and definitions.
\input{999blurb.tex}
Next, we briefly review PageRank and its incremental version in dynamic graphs.

\vspace{-2mm}
\paragraph{\bf PageRank}\label{sec:pagerank}
%PageRank (PR)~\cite{page1999pagerank} is a widely used algorithm to measure importance of nodes in a graph.
%The insight behind PageRank is that a node is important if it is linked to by many important nodes.
As shown in \cite{TPAYoonJK18}, PageRank scores for all nodes are represented as a PageRank score vector $\p$ which is defined by the following equation:
\small
\begin{equation*}
\p = (1-c)\sum_{i=0}^{\infty}(c\NAT)^i\b
\end{equation*}
\normalsize
where $c$ is a damping factor, $\NA$ is the row-normalized adjacency matrix, and $\b$ is the starting vector.
This equation is interpreted as a propagation of scores across a graph: 
initial scores in the starting vector $\b$ are propagated across the graph by multiplying with $\NAT$; 
since the damping factor $c$ is smaller than $1$, propagated scores converge, resulting in PageRank scores. 
As shown in \cite{yoon2018fast}, PageRank computation time is proportional to the L1 length of the starting vector $\b$, since small L1 length of $\b$ leads to faster convergence of iteration ($\sum_{i=0}^{\infty}(c\NAT)^i$) with damping factor $c$.
Here L1 length of a vector is defined as the sum of absolute values of its entries. 
The L1 length of a matrix is defined as the maximum L1 length of its columns. 
\vspace{-2mm}
\paragraph{\bf Incremental PageRank}\label{sec:incremental}
When edges are inserted or deleted, PageRank scores can be updated incrementally from the previous PageRank scores.
%by initializing it via a warm-start approach, which has small L1 length of the starting vector.
Let $\NA$ be the row-normalized adjacency matrix of a graph $G$ and $\NB$ be the row-normalized adjacency matrix after a change $\Delta G$ happened during $\dt$.
From now on, denote $\DA = \NBT - \NAT$, the difference between transpose of normalized matrices $\NAT$ and $\NBT$.
\begin{lemm}[Dynamic PageRank, Theorem 3.2 in~\cite{yoon2018fast}]
	\label{lemma:dynamic_ps}
	Given updates $\DA$ in a graph during $\dt$, an updated PageRank vector $\p(t+\dt)$ is computed incrementally from a previous PageRank vector $\p(t)$ as follows:
	\small
	\begin{equation*}
	\label{equ:update_ps}
	\p(t+\dt) = \p(t) + \sum_{k=0}^{\infty}(c(\NAT+\DA))^{k}c\DA\p(t)
	\end{equation*}
	\normalsize
\end{lemm}
\noindent Note that, for small changes in a graph, the L1 length of the starting vector $c\DA\p(t)$ is much smaller than $1$, the L1 length of the starting vector $\b$ in the static PageRank equation, resulting in much faster convergence.

%\vspace{3mm}
%PageRank has been used in various data mining tasks~\cite{chakrabarti2011index, fujiwara2012fast} including ranking~\cite{tong2008random}, community detection~\cite{zhu2013local, whang2013overlapping}, and link prediction~\cite{backstrom2011supervised}.
%While PageRank has found widespread utility, its 1st and 2nd derivatives have not been analyzed deeply, especially for anomaly detection.

\begin{table}[!t]
	\vspace{-5mm}
	\centering
	\small
	\caption{Table of symbols.}
	\begin{tabular}{cl}
		\toprule
		\textbf{Symbol} & \textbf{Definition} \\
		\midrule
		$G$ & {\small (un)directed and (un)weighted input graph} \\
		$\Delta G$ & {\small update in graph}\\
		$n,m$ & {\small numbers of nodes and edges in $G$}\\
		$\NA$ & {\small ($n \times n$) row-normalized adjacency matrix of $G$}\\
		$\NB$ & {\small ($n \times n$) row-normalized adjacency matrix of $G+\Delta G$}\\
		$\Delta \A$ & {\small($n \times n$) difference between $\NAT$ and $\NBT$ ($=\NBT-\NAT$)}\\	
		$c$ & {\small damping factor of PageRank}\\
		$\bs$ & {\small ($n\times1$) uniform starting vector}\\
		$\be$ & {\small ($n\times1$) out-edge proportional starting vector}\\
		$\NAS$ & {\small ($n \times n$) row-normalized unweighted adjacency matrix}\\
		$\NAE$ & {\small ($n \times n$) row-normalized weighted adjacency matrix}\\
		\bottomrule
	\end{tabular}
	\label{tab:symbols}
\end{table}

%% file: 999blurb.tex
Various node score functions have been designed to estimate importance (centrality, etc.) of nodes in a graph:
PageRank~\cite{page1999pagerank};
HITS~\cite{kleinberg1999authoritative}
and its derivatives (SALSA)~\cite{DBLP:journals/tois/LempelM01};
Fiedler vector~\cite{chungspectral};
all the centrality measures from social network analysis (eigenvector-, degree-, betweeness-centrality~\cite{wasserman1994social}).
Among them, we extend PageRank to design our node score functions in Section~\ref{sec:two_pr} because 
(a) it is fast to compute,
(b) it led to the ultra-successful ranking
method of Google,
(c) it is intuitive (`your importance depends on the importance of your neighbors').

%% file: 030proposed_method.tex
Node scores present the importance of nodes across a given graph.
Thus, as the graph evolves under normal behavior with the insertion and deletion of edges, node scores evolve smoothly.
In contrast, anomalies such as network attacks or rating manipulation often aim to complete their goals in a short time, e.g. to satisfy their clients, inducing large and abrupt changes in node scores.
Our key intuition is that such abrupt gains or losses are reflected in the {\bf 1st and 2nd derivative} of node scores: 
large 1st derivative identifies large changes, while large 2nd derivative identifies abrupt changes in the trend of the data, thereby distinguishing changes from normal users who evolve according to smooth trends.
Thus, tracking 1st and 2nd order derivatives helps detect anomalies in dynamic graphs.
%TODO revive the below sentence for PageRank version
%This PageRank-based approach contributes to an intuitive and computationally efficient anomaly detection algorithm which covers various types of anomalies in dynamic graphs.

Changes in a dynamic graph are classified into two types: structure changes and edge weight changes.
Since these two types of changes affect node scores of the graph differently (more details in Section~\ref{sec:two_pr}), we need to handle them separately.
Thus, first, we classify anomalies in dynamic graphs into two types: \structure and \edge (Section~\ref{sec:two_anomalies}).
Then we design two node score functions based on characteristics of these two types of anomalies, respectively (Section~\ref{sec:two_pr}).
Next, we define two novel metrics for anomalousness using 1st and 2nd order derivatives of our node scores, and verify the effectiveness of each metric on the respective type of anomalies theoretically (Section~\ref{sec:anomrank}).
Based on these analyses, we introduce our method \method, a fast and accurate anomaly detection algorithm in dynamic graphs (Section~\ref{sec:algorithm}).

\begin{figure}[!t]
	\centering
	\vspace{-5mm}
	\subfigure[Structure Change]
	{
		\label{fig:anomalyS}
		\includegraphics[width=.47\linewidth]{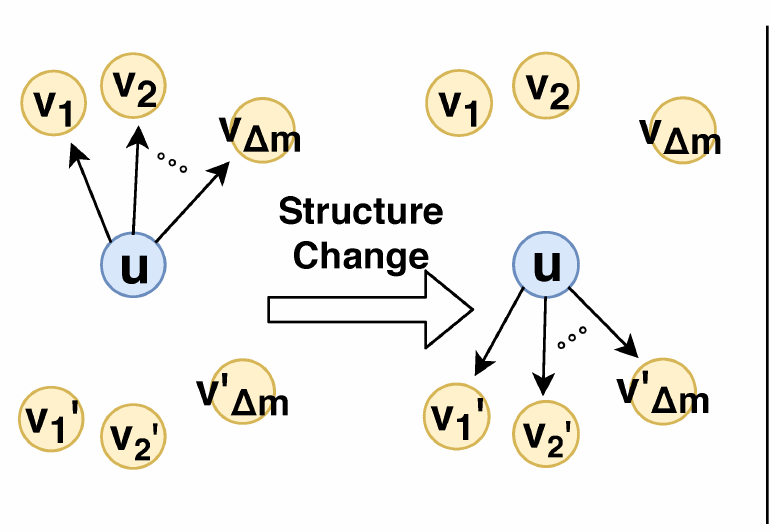}	
	}
	\subfigure[Edge Weight Change]
	{
		\label{fig:anomalyW}
		\includegraphics[width=.47\linewidth]{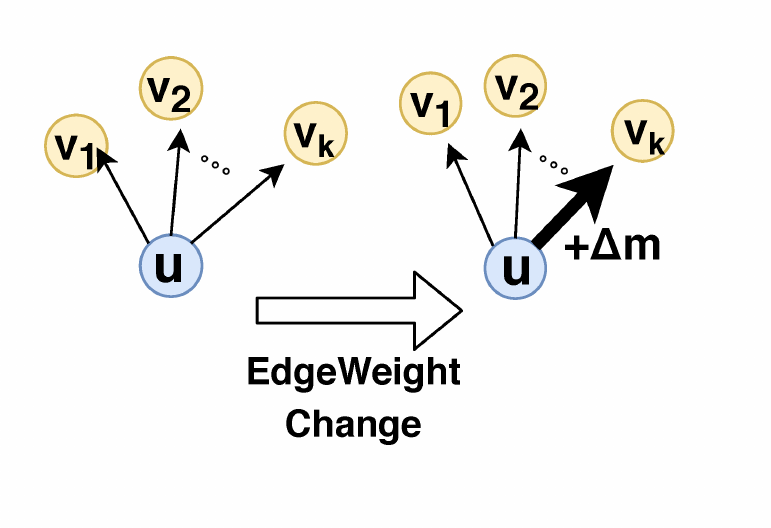}
	}
	\caption
	{ 
		Two-pronged approach:
		Changes in dynamic graphs are classified into two types, structure change and edge weight change.
	}
	\label{fig:anomalies}
\end{figure}

\subsection{Anomalies in Dynamic Graphs}
\label{sec:two_anomalies}
We classify anomalies in dynamic graphs into two types: \structure and \edge.

\subsubsection{\structure}
It is suspicious if a number of edges are inserted/deleted among nodes which are previously unrelated/related. 
Hence, \structure denotes a massive change with regard to the graph structure.
One example of \structure is spam in mail network graphs:
a spammer sends mail to many unknown individuals, generating out-edges toward previously unrelated nodes.
Data exfiltration attacks in computer network graphs are another example:
attackers transfer a target's data stealthily, generating unseen edges around a target machine to steal information.
As illustrated in Figure~\ref{fig:anomalies}, we define a structure change as follows:

\begin{defi}[Structure Change]
	\label{def:structure}
	If a node $u$ changes the destination of $\Delta m$ of its out-edges from previous neighbors $v_1, \dots, v_{\Delta m}$ to new neighbors $v_1', \dots, v_{\Delta m}'$, we call the change a {\it structure change} of size $\Delta m$.
\end{defi}
\noindent With abnormally large $\Delta m$, a structure change becomes an \structure.
To detect \structure, we need to focus on the existence of edges between two nodes, rather than the number of occurrences of edges between two nodes.

\subsubsection{\edge}
In dynamic graphs, an edge between two nodes could occur several times.
Edge weight is proportional to the number of edge occurrences.
\edge denotes a massive change of edge weights in a graph.
One example of \edge is port scan attacks in computer network graphs:
to scan ports in a target IP address, attackers repeatedly connect to the IP address, thus increasing the number of edge occurrences to the target node.
On Twitter, high edge density on a user-keyword graph could indicate bot-like behavior, e.g. bots posting about the same content repeatedly.
As illustrated in Figure~\ref{fig:anomalies}, we define an edge weight change as follows:

\begin{defi}[Edge Weight Change]
	\label{def:weight}
	If a node $u$ adds/subtracts $\Delta m$ out-edges to neighbor node $v$, we call the change an {\it edge weight change} of size $\Delta m$.
\end{defi}

\noindent With abnormally large $\Delta m$, an edge weight change becomes an \edge.
In contrast to \structure, here we focus on the number of occurrences of each edge, rather than only the presence or absence of an edge.

\subsection{Node Score Functions for Detecting \structure and \edge}
\label{sec:two_pr}
To detect \structure and \edge, we first define two node score functions, \prS and \prE, which we use to define our anomalousness metrics in Section~\ref{sec:anomrank}.

\subsubsection{\prS}
We introduce node score \prS, which we use to catch \structure.
Define the row-normalized unweighted adjacency matrix $\NAS$, a starting vector $\bs$ which is an all-$\frac{1}{n}$ vector of length $n$ (the number of nodes), and the damping factor $c$.
\begin{defi}[\prS]
	\prS node score vector $\ps$ is defined by the following iterative equation:
	\small
	\begin{equation*}
		\ps = c\NATS\ps + (1-c)\bs
	\end{equation*}
	\normalsize
\end{defi}
\noindent For this (unweighted) case, \prS is the same as PageRank, but we refer to it as \prS for consistency with our later definitions. 
Note that the number of edge occurrences between nodes is not considered in \prS. 
Using Lemma \ref{lemma:dynamic_ps}, we can compute $\prS$ incrementally at fast speed, in dynamic graphs. 

\subsubsection{\prE}
Next, we introduce the second node score, \prE, which we use to catch \edge. 
To incorporate edge weight, we use the weighted adjacency matrix $\AE$ instead of $\AS$. 
However, this is not enough on its own: imagine an attacker node who adds a massive number of edges, all toward a single target node, and the attacker has no other neighbors. 
Since $\NAE$ is row-normalized, this attacker appears no different in $\NAE$ as if they only added a single edge toward the same target. 
Hence, to catch such attackers, we also introduce an \emph{out-degree proportional starting vector} $\be$, i.e. setting the initial scores of each node proportional to its outdegree.

\begin{defi}[\prE]
	%Given a row-normalized weighted adjacency matrix $\NAE$ and an out-degree proportional starting vector $\be$, 
	\prE node score vector $\pe$ is defined by the following iterative equation:
	\small
	\begin{equation*}
	\pe = c\NATE\pe + (1-c)\be
	\end{equation*}
	\normalsize
\end{defi}
\noindent $\AE(i,j)$ is the edge weight from node $i$ to node $j$.
$\be(i)$ is $\frac{m_i}{m}$, where $m_i$ denotes the total edge weight of out-edges of node $i$, and $m$ denotes the total edge weight of the graph.

Next, we show how \prE is computed incrementally in a dynamic graph.
Assume that a change $\Delta G$ happens in graph $G$ in time interval $\dt$, inducing changes $\DAE$ and $\bed$ in the adjacency matrix and the starting vector, respectively.
\begin{lemm}[Dynamic \prE]
	\label{lemma:dynamic_pe}
	Given updates $\DAE$ and $\bed$ in a graph during $\dt$, an updated score vector $\pe(t+\dt)$ is computed incrementally from a previous score vector $\pe(t)$ as follows:
	\vspace{-2mm}
	\small
	\begin{align*}
	\label{equ:update_pe}
	\pe(t+\dt) =~&\pe(t) + \sum_{k=0}^{\infty}(c(\NATE+\DAE))^{k}c\DAE\pe(t) \\
	&+  (1-c)\sum_{k=0}^{\infty}(c(\NATE+\DAE))^{k}\bed
	\end{align*}
	\normalsize
	\vspace{-2mm}
	\begin{proof}
		For brevity, $\pen \leftarrow \pe(t+\dt)$ and $\peo \leftarrow \pe(t)$.
		\small
		\begin{align*}
		&\pen = (1-c)\sum_{k=0}^{\infty}c^k(\NATE+\DAE)^k(\be + \bed)\\
		&= (1-c)\sum_{k=0}^{\infty}c^k(\NATE+\DAE)^k\be + (1-c)\sum_{k=0}^{\infty}c^k(\NATE+\DAE)^k\bed\\
		&= \pe^{o} + \sum_{k=0}^{\infty}(c(\NATE+\DAE))^{k}c\DAE\pe^{o} +  (1-c)\sum_{k=0}^{\infty}(c(\NATE+\DAE))^{k}\bed
		\end{align*}
		\normalsize
		%In the first line, we use an equivalent form of PageRank described in our preliminary.
		In the third line, we use Lemma~\ref{lemma:dynamic_ps}.
	\end{proof}
\end{lemm}
\noindent Note that, for small changes in a graph, the starting vectors of the last two terms, $c\DAE\pe(t)$ and $\bed$ have much smaller $L1$ lengths than the original starting vector $\be$, so they can be computed at fast speed.

\subsubsection{Suitability}
\label{sec:suitable}

We estimate changes in \prS induced by a structure change (Definition~\ref{def:structure}) and compare the changes with those in \prE to prove the suitability of \prS for detecting \structure.

\begin{lemm}[Upper Bound for Structure Change in \prS]
	\label{lemma:ps_structure}
	When a structure change of size $\Delta m$ happens around a node $u$ with $k$ out-neighbors, $\lVert\DAS\rVert_{1}$ is upper-bounded by $\frac{2\Delta m}{k}$.
	\begin{proof}
		In $\DAS$, only the $u$-th column has nonzeros.
		Thus, $\lVert\DAS\rVert_{1}$
		$=\lVert\DAS(u)\rVert_{1}$.
		$\DAS(u)$ is normalized by $k$ as the total number of out-neighbors of node $u$ is $k$.
		For out-neighbors $v_i = v_1, \dots, v_{\Delta m}$ who lose edges, $\DAS(v_i,u) = -\frac{1}{k}$.
		For out-neighbors $v_i' = v_1', \dots, v_{\Delta m}'$ who earn edges, $\DAS(v_i',u) = \frac{1}{k}$.
		Then $\lVert\DAS\rVert_{1}=\lVert\DAS(u)\rVert_{1} = \frac{\Delta m}{k}+\frac{\Delta m}{k} = \frac{2\Delta m}{k}$.
	\end{proof}
\end{lemm}
\noindent When a structure change is presented in \prE, $\lVert\Delta \be \rVert_{1} = 0$ since there is no change in the number of edges.
Moreover $\lVert\DAE\rVert_{1}=\frac{2\Delta m}{m_u}$ since each row in $\AE$ is normalized by the total sum of out-edge weights, $m_u$, which is larger than the total number of out-neighbors $k$.
In other words, a structure change generates larger changes in \prS ($\frac{2\Delta m}{k}$) than \prE ($\frac{2\Delta m}{m_u}$).
Thus \prS is more suitable to detect \structure than \prE.

Similarly, we estimate changes in \prE induced by an edge weight change (Definition~\ref{def:weight}) and compare the changes with those in \prS to prove the suitability of \prE for detecting \edge.

\begin{lemm}[Upper Bound for Edge Weight Change in \prE]
	\label{lemma:pe_edge}
	When an edge weight change of size $\Delta m$ happens around a node $u$ with $m_u$ out-edge weights in a graph with $m$ total edge weights, $\lVert\DAE\rVert_{1}$ and $\lVert\bed\rVert_{1}$ are upper bounded by $\frac{2\Delta m}{m_u}$ and $ \frac{2\Delta m}{m}$, respectively.
%	\begin{align*}
%	\lVert\DAE\rVert_{1} \le \frac{2\Delta m}{m_u} \text{  and  } \lVert\bed\rVert_{1} \le \frac{2\Delta m}{m}. 
%	\end{align*}
	\begin{proof}
		In $\DAE$, only the $u$-th column has nonzeros.
		Then $\lVert\DAE\rVert_{1}$
		$=\lVert\DAE(u)\rVert_{1}$.
		Node $u$ has $m_{v_i}$ edges toward each out-neighbor $v_i (i=1,\dots,k)$.
		Thus the total sum of out-edge weights, $m_u$, is $\sum_{i=1}^{k}m_{v_i}$.
		Since an weighted adjacency matrix is normalized by the total out-edge weights, $\NATE(v_i,u) = \frac{m_{v_i}}{m_u}$.
		After $\Delta m$ edges are added from node $u$ to node $v_k$, $\DAE(v_i,u) = \frac{m_{v_i}}{m_u+\Delta m} - \frac{m_{v_i}}{m_u}$ for $i \neq k$, $\DAE(v_i,u) = \frac{m_{v_i}+\Delta m}{m_u+\Delta m} - \frac{m_{v_i}}{m_u}$ for $i=k$.
		Then $\lVert\DAE\rVert_{1}=\lVert\DAE(u)\rVert_{1}$ is bounded as follows:
		\vspace{-1mm}
		\small
		\begin{align*}
		\lVert\DAE\rVert_{1}=\lVert\DAE(u)\rVert_{1} &= \sum_{i=1}^{k}m_{v_i}(\frac{1}{m_u}-\frac{1}{m_u+\Delta m}) + \frac{\Delta m}{m_u + \Delta m}\\
		&=  \frac{2\Delta m}{m_u + \Delta m} \le \frac{2\Delta m}{m_u}
		\end{align*} 
		\normalsize
		$\be(i) = \frac{m_i}{m}$ where $m_i$ is the total sum of out-edge weights of node $i$.
		After $\Delta m$ edges are added from node $u$ to node $v_k$, $\bed(i) = \frac{m_i}{m+\Delta m} - \frac{m_i}{m}$ for $i \neq u$, $\bed(i) = \frac{m_i + \Delta m}{m+\Delta m} - \frac{m_i}{m}$ for $i=u$.
		Then $\lVert\bed\rVert_{1}$ is bounded as follows:
		\vspace{-1mm}
		\small
		\begin{align*}
		\lVert\bed\rVert_{1} = \sum_{i=1}^{n}m_i(\frac{1}{m}-\frac{1}{m+\Delta m}) + \frac{\Delta m}{m + \Delta m} = \frac{2\Delta m}{m + \Delta m} \le \frac{2\Delta m}{m}
		\end{align*}
		\normalsize
	\end{proof}
\end{lemm}
\vspace{-2mm}
\noindent In contrast, when an edge weight change is presented in \prS, $\lVert\DAS\rVert_{1}=0$ since the number of out-neighbors is unchanged.
Note that $\lVert\Delta \bs \rVert_{1} = 0$ since $\bs$ is fixed in \prS.
In other words, \edge does not induce any change in \prS.

\subsection{Metrics for \structure and \edge}
\label{sec:anomrank}
Next, we define our two novel metrics for evaluating the anomalousness at each time in dynamic graphs.

\subsubsection{\anS}
\label{sec:anomrank:s}
First, we discretize the first order derivative of \prS vector $\ps$ as follows:
\small
\begin{equation*}
	\psone = \frac{\ps(t+\Delta t)-\ps(t)}{\dt}
\end{equation*}
\normalsize

\noindent Similarly, the second order derivative of $\ps$ is discretized as follows:
\small
\begin{align*}
\pstwo = \frac{(\ps(t+\Delta t)-\ps(t)) - (\ps(t)-\ps(t - \Delta t))}{\Delta t^2}
\end{align*}
\normalsize

\noindent Next, we define a novel metric \anS which is designed to detect \structure effectively.
\begin{defi}[\anS]
	Given \prS vector $\ps$, \anS $\as$ is an $(n\times2)$ matrix $[\ps'~\ps'']$, concatenating 1st and 2nd derivatives of $\ps$.
	The \anS score is $\lVert\as\rVert_{1}$.
\end{defi}

\noindent Next, we study how \anS scores change under the assumption of a normal stream, or an anomaly, thus explaining how it distinguishes anomalies from normal behavior.
First, we model a normal graph stream based on Lipschitz continuity to capture smoothness:
%A common way to model smoothness is based on Lipschitz continuity: 

\begin{assumption}[$\lVert\ps(t)\rVert_{1}$ in Normal Stream]
	\label{ob:s_normal}
	In a normal graph stream, $\lVert\ps(t)\rVert_{1}$ is Lipschitz continuous with positive real constants $K_1$ and $K_2$ such that,
	\begin{align*}
		\lVert\psone\rVert_{1} \le K_1~and~\lVert\pstwo\rVert_{1} \le K_2
	\end{align*}
\end{assumption}

\noindent In Lemma~\ref{lemma:ps_one}, we back up Assumption~\ref{ob:s_normal} by upper-bounding $\lVert\psone\rVert_{1}$ and $\lVert\pstwo\rVert_{1}$.
For brevity, all proofs of this subsection are given in Supplement~\ref{sec:app:proofs}.

\begin{lemm}[Upper bound of $\lVert\psone\rVert_{1}$]
	\label{lemma:ps_one}
	Given damping factor $c$ and updates $\DAS$ in the adjacency matrix during $\dt$, $\lVert\psone\rVert_{1}$ is upper-bounded by $\frac{c}{1-c}\lVert\frac{\DAS}{\dt}\rVert_{1}$.
	\begin{proof}
		Proofs are given in Supplement~\ref{sec:app:proofs}.
	\end{proof}
\end{lemm}

\noindent We bound the $L1$ length of $\pstwo$ in terms of $L1$ length of $\DASO$ and $\DASN$, where $\DASO$ denotes the changes in $\AS$ from time $(t-\Delta t)$ to time $t$, and $\DASN$ denotes the changes in $\AS$ from $t$ to $(t+\Delta t)$.

\begin{lemm}[Upper bound of $\lVert\pstwo\rVert_{1}$]
	\label{lemma:ps_two}
	Given damping factor $c$ and sequencing updates $\DASO$ and $\DASN$,
	$\lVert\pstwo\rVert_{1}$ is upper-bounded by $\frac{1}{\dt^2}(\frac{c}{1-c}\lVert\DASN - \DASO\rVert_{1} + (\frac{c}{1-c})^2\lVert\DASN\rVert_{1} ^2 + (\frac{c}{1-c})^2\lVert\DASO\rVert_{1}^2)$.
	\begin{proof}
		Proofs are given in Supplement~\ref{sec:app:proofs}.
	\end{proof}
\end{lemm}

\noindent Normal graphs have small changes thus having small $\lVert\DAS\rVert_{1}$.
This results in small values of $\lVert\psone\rVert_{1}$.
In addition, normal graphs change gradually thus having small$\lVert\DASN - \DASO\rVert_{1}$.
This leads to small values of $\lVert\pstwo\rVert_{1}$.
Then, \anS score $\lVert\as\rVert_{1} = \max(\lVert\psone\rVert_{1}, \lVert\pstwo\rVert_{1})$ has small values in normal graph streams under small upper bounds.

\begin{observation}[\structure in \anS]
	\label{ob:s_anomaly}
	\structure involves sudden structure changes, inducing large \anS scores.
\end{observation}

\noindent \structure happens with massive changes ($\dm$) abruptly ($\dmm$).
In the following Theorem~\ref{theorem:up_d1d2_s}, we explain Observation~\ref{ob:s_anomaly} based on large values of $\dm$ and $\dmm$ in \structure.

\begin{theo}[Upper bounds of $\lVert\psone\rVert_{1}$ and $\lVert\pstwo\rVert_{1}$ with \structure]
	\label{theorem:up_d1d2_s}
When \structure occurs with large $\dm$ and $\dmm$, L1 lengths of $\psone$ and $\pstwo$ are upper-bounded as follows:
	\small
	\begin{empheq}[box=\widefbox]{align}
		\lVert\psone\rVert_{1} &\le \frac{c}{1-c}\frac{2}{k}\dm \nonumber \\
		\lVert\pstwo\rVert_{1} &\le \frac{c}{1-c}\frac{2}{k}\dmm + 2(\frac{c}{1-c})^2(\frac{2}{k})^2(\dm)^2 \nonumber
	\end{empheq}
	\normalsize
	\begin{proof}
		Proofs are given in Supplement~\ref{sec:app:proofs}.
	\end{proof}
\end{theo}

\noindent  
Based on Theorem~\ref{theorem:up_d1d2_s}, \structure has higher upper bounds of $\lVert\psone\rVert_{1}$ and $\lVert\pstwo\rVert_{1}$ than normal streams.
This gives an intuition for why \structure results in high \anS scores (Figure~\ref{fig:perf:anomaly_score}).
We detect \structure successfully based on \anS scores in real-world graphs (Figure~\ref{fig:accuracy}).

\subsubsection{\anE}
We discretize the first and second order derivatives $\peone$ and $\petwo$ of $\pe$ as follows:
\small
\begin{align*}
	\peone &= \frac{\pe(t+\Delta t)-\pe(t)}{\dt}\\
	\petwo &= \frac{(\pe(t+\Delta t)-\pe(t)) - (\pe(t)-\pe(t - \Delta t))}{\Delta t^2}
\end{align*}
\normalsize

\noindent Then we define the second metric \anE which is designed to find \edge effectively.
\begin{defi}[\anE]
	Given \prE vector $\pe$, \anE $\ae$ is a $(n\times2)$ matrix $[\pe'~\pe'']$, concatenating 1st and 2nd derivatives of $\pe$.
	The \anE score is $\lVert\ae\rVert_{1}$.
\end{defi}

\noindent We model smoothness of $\lVert\pe(t)\rVert_{1}$ in a normal graph stream using Lipschitz continuity in Assumption~\ref{ob:e_normal}.
Then, similar to what we have shown in the previous Section~\ref{sec:anomrank:s}, we show upper bounds of $\lVert\peone\rVert_{1}$ and $\lVert\petwo\rVert_{1}$ in Lemmas~\ref{lemma:pe_one} and ~\ref{lemma:pe_two} to explain Assumption~\ref{ob:e_normal}.

\begin{assumption}[$\lVert\pe(t)\rVert_{1}$ in Normal Stream]
	\label{ob:e_normal}
	In a normal graph stream, $\lVert\pe(t)\rVert_{1}$ is Lipschitz continuous with positive real constants $C_1$ and $C_2$ such that,
	\begin{align*}
	\lVert\peone\rVert_{1} \le C_1~and~\lVert\petwo\rVert_{1} \le C_2
	\end{align*}
\end{assumption}

\begin{lemm}[Upper bound of $\lVert\peone\rVert_{1}$]
	\label{lemma:pe_one}
	Given damping factor $c$, updates $\DAE$ in the adjacency matrix, and updates $\bed$ in the starting vector during $\dt$, $\lVert\peone\rVert_{1}$ is upper-bounded by $\frac{1}{\dt}(\frac{c}{1-c}\lVert\DAE\rVert_{1}+\lVert\bed\rVert_{1})$.
	\begin{proof}
		Proofs are given in Supplement~\ref{sec:app:proofs}.
	\end{proof}
\end{lemm}

\noindent In the following lemma, $\lVert\pstwo\rVert_{max}$ denotes the upper bound of $\lVert\pstwo\rVert_{1}$ which we show in Lemma~\ref{lemma:ps_two}.
$\DAEO$ is the changes in $\AE$ from time $(t-\Delta t)$ to time $t$, and $\DAEN$ is the changes in $\AE$ from $t$ to $(t+\Delta t)$.
$\bedo$ is the changes in $\be$ from time $(t-\Delta t)$ to time $t$, and $\bedn$ is the changes in $\be$ from $t$ to $(t+\Delta t)$.

\begin{lemm}[Upper bound of $\lVert\petwo\rVert_{1}$]
	\label{lemma:pe_two}
	Given damping factor $c$, sequencing updates $\DAEO$ and $\DAEN$, and sequencing updates $\bedo$ and $\bedn$,
	the upper bound of $\lVert\petwo\rVert_{1}$ is presented as follows:\\
	\small
	\begin{align*}
		\lVert\pstwo\rVert_{max} + \frac{1}{\dt^2}(\lVert\bedn-\bedo \rVert_{1} + \frac{c}{1-c}\lVert\DAEN\rVert_{1}\lVert\bedn\rVert_{1})
	\end{align*}
	\normalsize
	\begin{proof}
		Proofs are given in Supplement~\ref{sec:app:proofs}.
	\end{proof}
\end{lemm}

\noindent Normal graph streams have small changes (small $\lVert\DAE\rVert_{1}$ and small $\lVert\bed\rVert_{1}$) and evolve gradually (small $\lVert\bedn-\bedo\rVert_{1}$). %without abrupt changes.
Then, normal graph streams have small \anE scores under small upper bounds of $\lVert\peone\rVert_{1}$ and $\lVert\petwo\rVert_{1}$.

\begin{observation}[\edge in \anE]
	\label{ob:e_anomaly}
	\edge involves sudden edge weight changes, inducing large \anE.
\end{observation}

\noindent We explain Observation~\ref{ob:e_anomaly} by showing large upper bounds of $\lVert\peone\rVert_{1}$ and $\lVert\petwo\rVert_{1}$ induced by large values of $\dm$ and $\dmm$ in \edge.

\begin{theo}[Upper bounds of $\lVert\peone\rVert_{1}$ and $\lVert\petwo\rVert_{1}$ with \edge]
	\label{theorem:up_d1d2_e}
	When \edge occurs with large $\dm$ and $\dmm$, L1 lengths of $\peone$ and $\petwo$ are upper-bounded as follows:
	\small
	\begin{empheq}[box=\widefbox]{align}
		\lVert\peone\rVert_{1} &\le \frac{c}{1-c}\frac{2}{m_u}\dm + \frac{2}{m}\dm \nonumber\\
		\lVert\petwo\rVert_{1} &\le \frac{c}{1-c}\frac{2}{k}\dmm + \frac{2}{m}\dmm \nonumber\\
		&+ 2(\frac{c}{1-c})^2(\frac{2}{k})^2(\dm)^2 + \frac{c}{1-c}\frac{2}{m_u}\frac{2}{m}(\dm)^2 \nonumber
	\end{empheq}
	\normalsize
	\begin{proof}
		Proofs are given in Supplement~\ref{sec:app:proofs}.
	\end{proof}
\end{theo}

\noindent  
With high upper bounds of $\lVert\peone\rVert_{1}$ and $\lVert\petwo\rVert_{1}$, shown in Theorem~\ref{theorem:up_d1d2_e}, \edge has high \anE scores (Figure~\ref{fig:perf:anomaly_score}).
We detect \edge successfully based on \anE scores in real-world graphs (Figure~\ref{fig:accuracy}).

\begin{algorithm} [t!]
	\small
	\begin{algorithmic}[1]
		\caption{\method}
		\label{alg:main}
		\REQUIRE updates in a graph: $\Delta G$, previous \prSE: $\ps^{old}, \pe^{old}$
		\ENSURE anomaly score: $s_{\text{anomaly}}$, updated \prSE: $\ps^{new}, \pe^{new}$
		\STATE compute updates $\DAS, \DAE$ and $\bed$ \label{line:updateAb}
		\STATE compute $\ps^{new}$ and $\pe^{new}$ incrementally from $\ps^{old}$ and $\pe^{old}$ using $\DAS, \DAE$ and $\bed$ \label{line:updatep}
		\STATE $s_{\text{anomaly}} = \textit{ComputeAnomalyScore}(\ps^{new}, \pe^{new})$ \label{line:score}
		\BlankLine
		\RETURN $s_{\text{anomaly}}$
	\end{algorithmic}
\end{algorithm}
\begin{algorithm} [t!]
	\small
	\begin{algorithmic}[1]
		\caption{ComputeAnomalyScore}
		\label{alg:anom_score}
		\REQUIRE \prS and \prE vectors: $\ps, \pe$
		\ENSURE anomaly score: $s_{\text{anomaly}}$
		\STATE compute \anS $\as = [\psone~\pstwo]$\label{line:anoms}
		\STATE compute \anE $\ae = [\peone~\petwo]$\label{line:anomw}
		\STATE $s_{\text{anomaly}} = \lVert\a\rVert_{1} = \max(\lVert\as\rVert_{1}, \lVert\ae\rVert_{1})$\label{line:finalscore}
		\BlankLine
		\RETURN $s_{\text{anomaly}}$
	\end{algorithmic}
\end{algorithm}

\subsection{Algorithm}
\label{sec:algorithm}
Algorithm~\ref{alg:main} describes how we detect anomalies in a dynamic graph. 
%From a high-level, given new changes occurring in a graph, we update \prS and \prE correspondingly, compute the resulting 1st and 2nd derivatives, and use them to compute anomalousness of the updates.
We first calculate updates $\DAS$ in the unweighted adjacency matrix, updates $\DAE$ in the weighted adjacency matrix, and updates $\bed$ in the out-edge proportional starting vector (Line \ref{line:updateAb}).
These computations are proportional to the number of edge changes, taking a few milliseconds for small changes.
Then, \method updates \prS and \prE vectors using the update rules in Lemmas~\ref{lemma:dynamic_ps} and ~\ref{lemma:dynamic_pe} (Line \ref{line:updatep}). 
%The update rules are also proportional to the number of edge changes, reducing the computation load by $\frac{\Delta m}{m}$ from static computation methods.
Then \method calculates an anomaly score given \prS and \prE in Algorithm~\ref{alg:anom_score}.
\method computes \anS and \anE, and returns the maximum L1 length between them as the anomaly score. 
%As shown in Algorithms~\ref{alg:main} and~\ref{alg:anom_score}, \method is an intuitive algorithm with simple implementation compared to other competitors~\cite{beutel2013copycatch,eswaran2018sedanspot,eswaran2018spotlight,liu2017holoscope,shin2017densealert}.

\textbf{Normalization}:
As shown in Theorems~\ref{theorem:up_d1d2_s} and~\ref{theorem:up_d1d2_e}, the upper bounds of \anS and \anE are based on the number of out-neighbors $k$ and the number of out-edge weights $m_u$.
This leads to skew in anomalousness score distributions since many real-world graphs have skewed degree distributions.
Thus, we normalize each node's \anS and \anE scores by subtracting its mean and dividing by its standard deviation, which we maintain along the stream.
% To handle this, we maintain mean and variance along the graph streams and normalize each node's \anS and \anE scores.
%In addition, since \anS and \anE have different scales, we no- rmalize each vector to have the maximum score among nodes as $1$.

\textbf{Explainability and Attribution}:
\method explains the type of anomalies by comparing \anS and \anE:
higher scores of \anS suggest that \structure has happened, and vice versa. 
High scores of both metrics indicate a large edge weight change that also alters the graph structure.
Furthermore, we can localize culprits of anomalies by ranking \method scores of each node in the score vector, as computed in Lines \ref{line:anoms} and \ref{line:anomw} of Algorithm \ref{alg:anom_score}. 
We show this localization experimentally in Section~\ref{sec:exp:local}.
%The higher the rank, the more suspicious is the node.

\textbf{$\Delta t$ selection}:
Our analysis and proofs, hold for any value of $\Delta t$.
The choice of $\Delta t$ is outside the scope of this paper, and probably best decided by a domain expert:
large $\Delta t$ is suitable for slow ('low temperature') attacks;
small $\Delta t$ spots fast and abrupt attacks.
In our experiments, we chose $\Delta t = 1$ hour, and $1$ day, respectively, for a computer-network intrusion setting, and for a who-emails-whom network. 

%% file: 040experiments.tex
In this section, we evaluate the performance of \method compared to state-of-the-art anomaly detection methods on dynamic graphs.
We aim to answer the following questions:
\begin{itemize} [leftmargin=10pt]
	\item {\textbf{Q1. Practicality.}} 
	How fast, accurate, and scalable is \method compared to its competitors? (Section~\ref{sec:exp:prac})
	%Does it outperform state-of-the-art anomaly detection methods on dynamic graphs? 
	\item {\textbf{Q2. Effectiveness of two-pronged approach.}} 
	How do our two metrics, \anS and \anE, complement each other in real-world and synthetic graphs? (Section~\ref{sec:exp:two})
	%by detecting different types of anomalies 
	\item {\textbf{Q3. Effectiveness of two-derivatives approach.}}
	How do the 1st and 2nd order derivatives of \prS and \prE complement each other? (Section~\ref{sec:exp:eval})
	\item {\textbf{Q4. Attribution.}} 
	How accurately does \method localize culprits of anomalous events? (Section~\ref{sec:exp:local})
\end{itemize}

\subsection{Setup}
We use SedanSpot~\cite{eswaran2018sedanspot} and DenseAlert~\cite{shin2017densealert}, state-of-the-art anomaly detection methods on dynamic graphs as our baselines.
We use two real-world dynamic graphs, \textit{DARPA} and \textit{ENRON}, and one synthetic dynamic graph, \textit{RTM}, with two anomaly injection scenarios.
Anomalies are verified by comparing to manual annotations or by correlating with real-world events.
More details of experimental settings and datasets are described in Supplement~\ref{sec:app:exp} and~\ref{sec:app:dataset}.

\begin{figure}[!t]
	\centering
	\includegraphics[width=.8\linewidth]{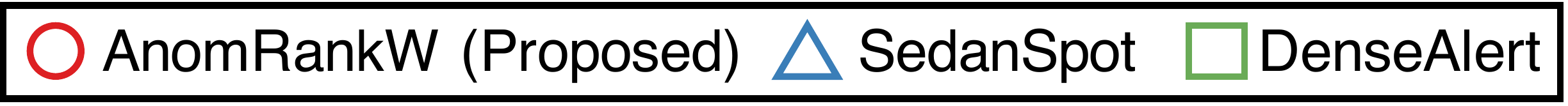}
	\subfigure[Precision@50 vs. time taken]
	{
		\label{fig:accuracy:2}
		\includegraphics[width=.37\linewidth]{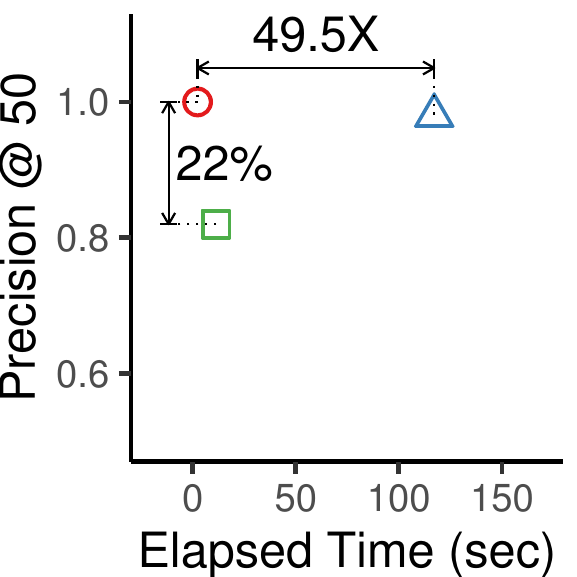}	
	}
	\subfigure[Precision@150 vs. time taken]
	{
		\label{fig:accuracy:3}
		\includegraphics[width=.37\linewidth]{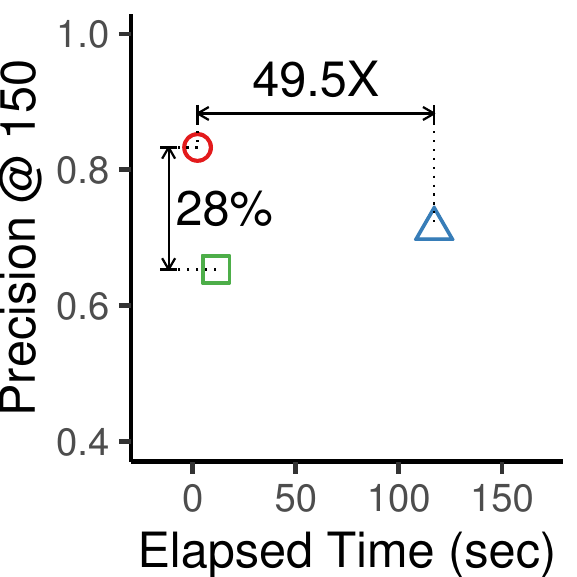}
	} 
	\subfigure[Precision@200 vs. time taken]
	{
		\label{fig:accuracy:4}
		\includegraphics[width=.37\linewidth]{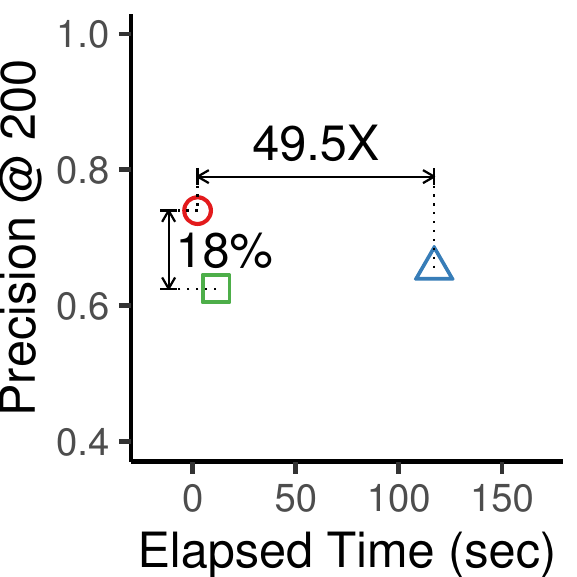}	
	}
	\subfigure[Precision vs. Recall]
	{
		\label{fig:accuracy:1}
		\includegraphics[width=.39\linewidth]{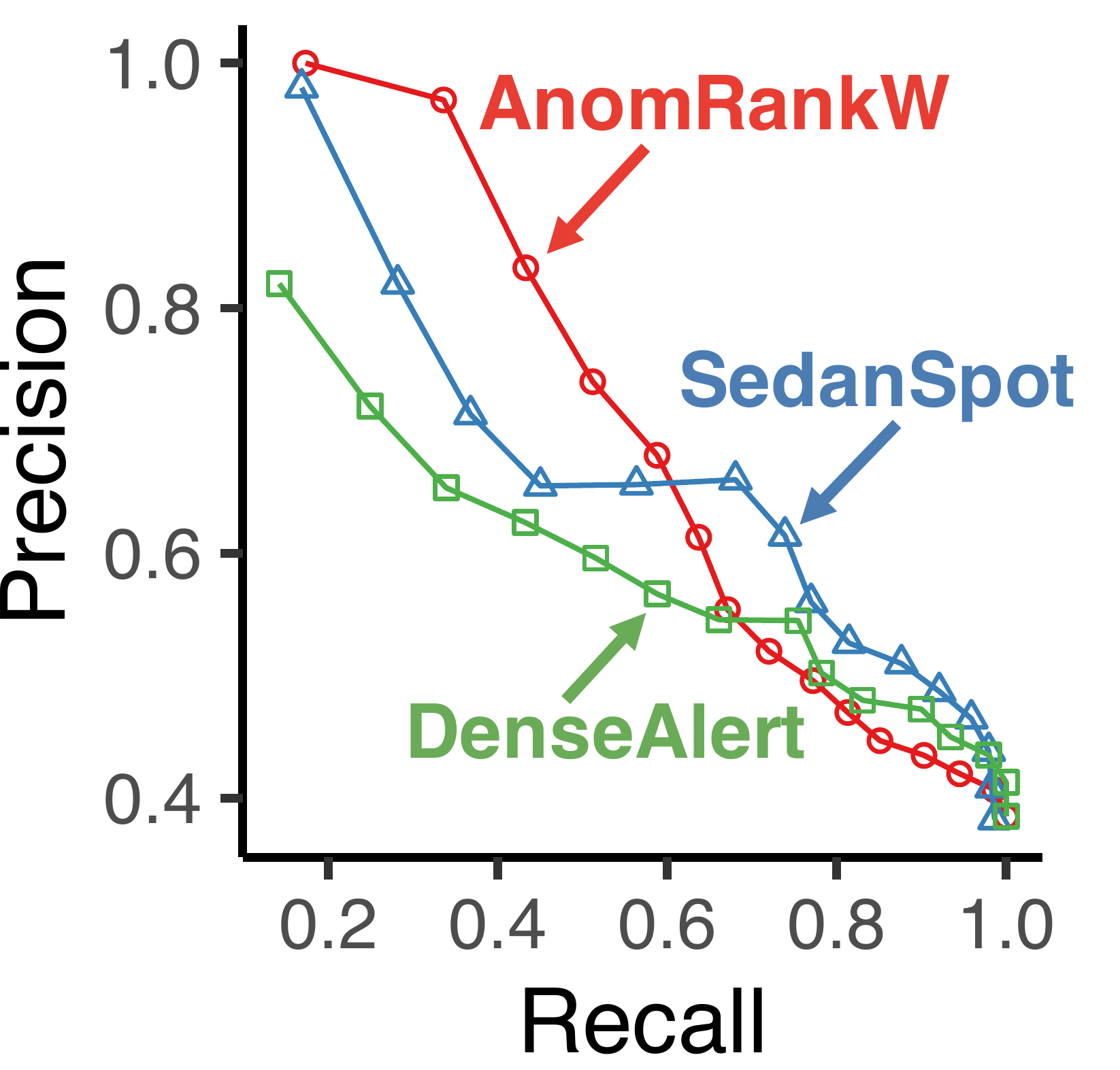}
	}
	\caption
	{ 
		\underline{\smash{\method is fastest and most accurate}}:
		(a-c) \method outperforms the baselines in terms of both accuracy and speed on the {\it DARPA} dataset.
		(d) \method achieves better precision and recall on high ranks (top-$50,\dots,250$) than its competitors.
	}
	\label{fig:accuracy}
\end{figure}

\subsection{Practicality}
\label{sec:exp:prac}
We examine the practicality of \method on the {\it DARPA} dataset, a public benchmark for Network Intrusion Detection Systems.
In this network intrusion setting, our focus is on detecting high-volume (i.e. high edge-weight) intrusions such as denial of service (DOS) attacks, which are typically the focus of graph mining-based detection approaches. 
Hence, we use only the \anE metric. 

\textbf{Precision and Recall:}
Using each method, we first compute anomaly scores for each of the $1139$ graph snapshots, then select the top-k most anomalous snapshots ($k=50,100,\cdots, 800$).
Then we compute precision and recall for each method's output. %compared to the ground truth annotations provided by the dataset. 
In Figure~\ref{fig:accuracy:1}, %each plotted point shows the precision and recall arising from one of these values of $k$. 
\method shows higher precision and recall than DenseAlert and SedanSpot on high ranks (top-$50,\dots,250$).
Considering that anomaly detection tasks in real-world settings are generally focused on the most anomalous instances, high accuracy on high ranks is more meaningful than high accuracy on low ranks.
Moreover, considering that the number of ground truth anomalies is $288$, its precision and recall up to top-$250$ better reflects its practicality.

\textbf{Accuracy vs. Running Time:}
In Figures~\ref{fig:perf:running_time} and~\ref{fig:accuracy}(a-c), \method is most accurate and fastest.
Compared to SedanSpot, \method achieves up to $18\%$ higher precision on top-k ranks with $49.5 \times$ faster speed. 
Compared to DenseAlert, \method achieves up to $35\%$ higher precision with $4.8\times$ faster speed.
DenseAlert and SedanSpot require several subprocesses (hashing, random-walking, reordering, sampling, etc), resulting in large computation time.

\textbf{Scalability:}
Figure~\ref{fig:perf:scalability} shows the scalability of \method with the number of edges.
We plot the wall-clock time needed to run on the (chronologically) first $2, 2^2, \dots, 2^{22}$ edges of the {\it DARPA} dataset. 
This confirms the linear scalability of \method with respect to the number of edges in the input dynamic graph.
Note that \method processes $1M$ edges within $1$ second, allowing real-time anomaly detection.

\textbf{Effectiveness:}
Figure~\ref{fig:perf:anomaly_score} shows changes of \method scores in the {\it DARPA} dataset, with time window of $\Delta T = 1$ hour.
Consistently with the large upper bounds shown in Theorems \ref{theorem:up_d1d2_s} and \ref{theorem:up_d1d2_e}, ground truth attacks (red crosses) have large \method scores in Figure~\ref{fig:perf:anomaly_score}.
Given {\it mean} and {\it std} of anomaly scores of all snapshots, setting an anomalousness threshold of ({\it mean + $\frac{1}{2}$std}), $77\%$ of spikes above the threshold are true positives.
This shows the effectiveness of \method as a barometer of anomalousness.

\subsection{Effectiveness of Two-Pronged Approach}
\label{sec:exp:two}
In this experiment, we show the effectiveness of our two-pronged approach using real-world and synthetic graphs.

\subsubsection{Real-World Graph}
We measure anomaly scores based on four metrics: \anS, \anE, SedanSpot, and DenseAlert, on the \textit{ENRON} dataset.
In Figure~\ref{fig:enron}, \anE and SedanSpot show similar trends, while \anS detects different events as anomalies on the same dataset.
DenseAlert shows similar trends with the sum of \anS and \anE, while missing several anomalous events. 
This is also reflected in the low accuracy of DenseAlert on the \textit{DARPA} dataset in Figure~\ref{fig:accuracy}.
The anomalies detected by \anS and \anE coincide with major events in the {\it ENRON} timeline\footnote{http://www.agsm.edu.au/bobm/teaching/BE/Enron/timeline.html} as follows:

\begin{enumerate} [leftmargin=10pt]
	\footnotesize
	\item June 12, 2000: Skilling makes joke at Las Vegas conference, comparing California to the Titanic.
	\item August 23, 2000: FERC orders an investigation into Timothy Belden's strategies designed to drive electricity prices up in California.
	\item Oct. 3, 2000: Enron attorney Richard Sanders travels to Portland to discuss Timothy Belden's strategies.
	\item Dec. 13, 2000: Enron announces that Jeffrey Skilling will take over as chief executive.
	\item Mar. 2001: Enron transfers large portions of EES business into wholesale to hide EES losses.
	\item July 13, 2001: Skilling announces desire to resign to Lay. Lay asks Skilling to take the weekend and think it over.
	\item Oct. 17, 2001: Wall Street Journal reveals the precarious nature of Enron's business. 
	%Sept. 26, 2001: Lay tells employees: Enron stock is an "incredible bargain". 
	%The SEC begins an informal probe of Enron.
	\item Nov. 19, 2001: Enron discloses it is trying to restructure a 690 million obligation.
	\item Jan. 23-25, 2002: Lay resigns as chairman and CEO of Enron. Cliff Baxter, former Enron vice chairman, commits suicide.
	\item Feb. 2, 2002: The Powers Report, a summary of an internal investigation into Enron's collapse spreads out.
\end{enumerate}
	\normalsize

\begin{figure}[!t]
	\includegraphics[width=1\linewidth]{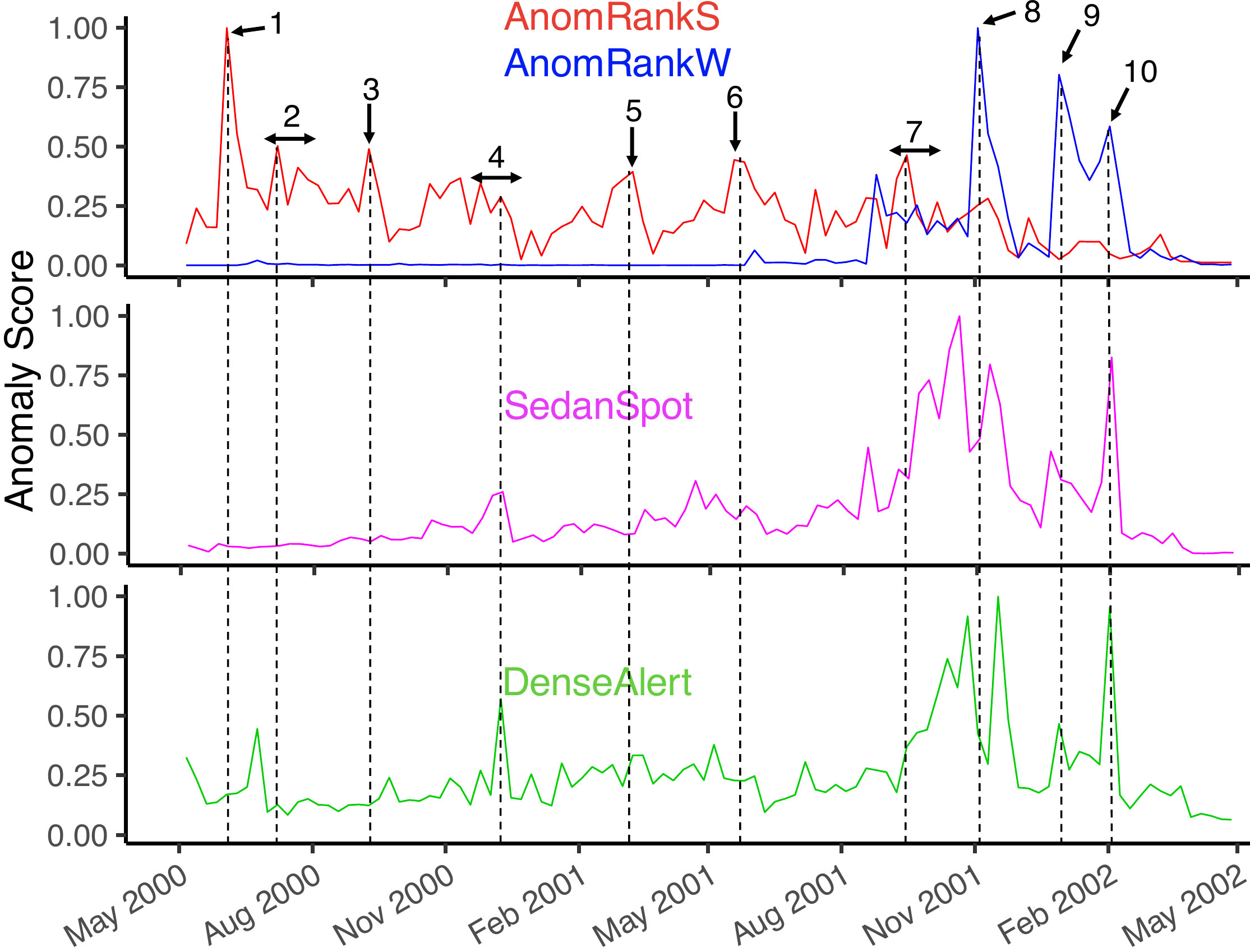}
	\caption
	{ 
		\label{fig:enron}
		\underline{\smash{Two-pronged approach pays off}}:
		\anE and SedanSpot show similar trends on the Enron dataset, while \anS detects different events as anomalies on the same dataset.
		DenseAlert shows similar trends with the sum of \anS and \anE while missing several anomalous events.
	}
\end{figure}

\noindent The high anomaly scores of \anS coincide with the timeline events when Enron conspired to drive the California electricity price up (Events 1,2,3) and office politics played out in Enron (Events 4,5,6).
Meanwhile, \anE shows high anomaly scores when the biggest scandals of the company continued to unfold (Events 7, 8,9,10).
Note that \anS and \anE are designed to detect different properties of anomalies on dynamic graphs.
\anS is effective at detecting \structure like unusual email communications for conspiration, while \anE is effective at detecting \edge like massive email communications about big scandals that swept the whole company.
The non-overlapping trends of \anS and \anE show that the two metrics complement each other successfully in real-world data.
Summarizing the two observations we make:
\begin{itemize} [leftmargin=10pt]
	\item \textbf{Observation 1}.
		\anS and \anE spot different types of anomalous events.
	\item \textbf{Observation 2}.
		DenseAlert and SedanSpot detect a subset of the anomalies detected by \method.
\end{itemize}

\subsubsection{Synthetic Graph} \label{sec:synthetic}

In our synthetic graph generated by \emph{RTM} method, we inject two types of anomalies to examine the effectiveness of our two metrics.
Details of the injections are as follows:

\begin{itemize} [leftmargin=10pt]
	\item \textbf{\sceneS}: We choose $50$ timestamps uniformly at random: at each chosen timestamp, we select $8$ nodes uniformly at random, and introduce all edges between these nodes in both directions. 
	\item \textbf{\sceneE}: We choose $50$ timestamps uniformly at random: at each chosen timestamp, we select two nodes uniformly at random, and add $70$ edges from the first to the second. 
\end{itemize}

\noindent A clique is an example of \structure with unusual structure pattern, while high edge weights are an example of \edge.
Hence, \sceneS and \sceneE are composed of \structure and \edge, respectively.

Then we evaluate the precision of the top-50 highest anomaly scores output by the \anS metric and the \anE metric.
We also evaluate each metric on the \emph{DARPA} dataset based on their top-250 anomaly scores. 
In Table~\ref{tab:metric_se}, \anS shows higher precision on \sceneS than \anE, while \anE has higher precision on \sceneE and \emph{DARPA}.
In Section~\ref{sec:suitable}, we showed theoretically that \structure induces larger changes in \anS than \anE, explaining the higher precision of \anS than \anE on \sceneS.
We also showed that adding additional edge weights has no effect on \anS, explaining that \anS does not work on \sceneE.
For the {\it DARPA} dataset, \anE shows higher accuracy than \anS.
{\it DARPA} contains 2.7M attacks, and $90\%$ of the attacks (2.4M attacks) are DOS attacks generated from only 2-3 source IP adderesses toward 2-3 target IP addresses.
These attacks are of \edge type with high edge weights.
Thus \anE shows higher precision on {\it DARPA} than \anS.
 
 \begin{table}[]
 	\centering
 	\caption{
 		\underline{\smash{\anS and \anE complement each other}}:\\
 		we measure precision of the two metrics on the synthetic graph with two injection scenarios (\sceneS, \sceneE) and the real-world graph \textit{DARPA}.
 		We measure precision on top-50 and top-250 ranks on the synthetic graph and \textit{DARPA}, respectively.
 	}
 	\label{tab:metric_se}
 	\begin{tabular}{c|c|c|c}
 		\hline	
 		\diagbox[width=2.8cm, height=0.8cm]{Metric}{Dataset} & \sceneS & \sceneE & {\it DARPA} \\ \hline
 		$\anS$ only& \textbf{.96} & .00 & .42 \\ \hline
 		$\anE$ only& .82 & \textbf{.79} & \textbf{.69} \\ \hline
 	\end{tabular}
 \end{table}
 
 \begin{table}[]
 	\centering
 	\caption{
 		\underline{\smash{1st and 2nd order derivatives complement each other}}:\\
 		\anEF and \anES are 1st and 2nd derivatives of \prE.
 		Combining \anEF and \anES results in the highest precision.
 		%We measure top-50 precision of each method on the synthetic graphs and top-200 on \textit{DARPA}.
 		%(a) \anEF and \anES show complementary behaviors;
 		%(b) combination of \anEF and \anES shows better precision.
 	}
 	\label{tab:metric_12}
 	\begin{tabular}{c|c|c|c}
 		\hline	
 		\diagbox[width=2.8cm, height=0.8cm]{Metric}{Dataset} & \sceneS & \sceneE & {\it DARPA} \\ \hline
 		\anEF& .06 & .11 & .65 \\
 		\anES& .80 & .78 & .61 \\ \hline
 		\anE & {\bf .82} & {\bf .79} & {\bf .69} \\ \hline
 	\end{tabular}
 \end{table}
 
\subsection{Effectiveness of Two-Derivatives Approach}
\label{sec:exp:eval}
In this experiment, we show the effectiveness of 1st and 2nd order derivatives of \prS and \prE in detecting anomalies in dynamic graphs.
For brevity, we show the result on \prE; result on \prS is similar.
Recall that \anE score is defined as the L1 length of $\ae = [\peone~\petwo]$ where $\peone$ and $\petwo$ are the 1st and 2nd order derivatives of \prE, respectively.
We define two metrics, \anEF and \anES, which denote the L1 lengths of $\peone$ and $\petwo$, respectively.
By estimating precision using \anEF and \anES individually, we examine the effectiveness of each derivative using the same injection scenarios and evaluation approach as Section \ref{sec:synthetic}.

In Table~\ref{tab:metric_12}, \anEF shows higher precision on the \textit{DARPA} dataset, while \anES has higher precision on injection scenarios.
\anEF detects suspiciously large anomalies based on L1 length of 1st order derivatives, while \textsc{AnomRa} \textsc{nkW-2nd} detects abruptness of anomalies based on L1 length of 2nd order derivatives.
%Since the injection scenarios have small size of anomalies (a clique of 8 nodes, or edge weight of 70) compared to the graph size (1K nodes and 8.1K edges), the precision of \anEF is lower than that of \anES.
%{\color{red}In contrast, \textit{DARPA} has large size of anomalies (NEED ANALYSIS).}
Note that combining 1st and 2nd order derivatives leads to better precision.
This shows that 1st and 2nd order derivatives complement each other.

\subsection{Attribution}
\label{sec:exp:local}
In this experiment, we show that \method successfully localizes the culprits of anomalous events as explained in the last paragraph of Section~\ref{sec:algorithm}. 
In Figure~\ref{fig:local}, given a graph snapshot detected as an anomaly in the {\it DARPA} dataset, nodes (IP addresses) are sorted in order of their \method scores.
Outliers with significantly large scores correspond to IP addresses which are likely to be engaged in network intrusion attacks.
At the $15$th snapshot ($T=15$) when {\it Back} DOS attacks occur, the attacker IP ($135.008.060.182$) and victim IP ($172.016.114.050$) have the largest \method scores.
In the $133$th snapshot ($T=133$) where {\it Nmap} probing attacks take place, the victim IP ($172.016.112.050$) has the largest score.
%This shows that \method can be easily extended to localize the anomalous node or subgraph of nodes based on the ranking of \method scores among nodes.

%% file: 060conclusion.tex
In this paper, we proposed a two-pronged approach for detecting anomalous events in a dynamic graph. 
% We classified anomalies into two types: \structure and \edge. 
% Our two novel metrics (\anS and \anE) track 1st and 2nd derivatives of \prS and \prE, thus detecting these anomalies effectively.
%Future work could further explore the connection between derivatives of PageRank and anomaly detection, both theoretically and practically, as well as in richer datasets such as attributed graphs. 

Our main contributions are:
\begin{itemize}[leftmargin=10pt]
	\item 
	{\textbf{Online, Two-Pronged Approach}
		We introduced \method, a novel and simple detection method in dynamic graphs.
	}
	\item {
	\textbf{Theoretical Guarantees}
	We present theoretical analysis (Theorems \ref{theorem:up_d1d2_s} and \ref{theorem:up_d1d2_e}) on the effectiveness of \method.
	}
	\item {
	\textbf{Practicality}
	In Section \ref{sec:experiments}, we show that \method outperforms state-of-the-art baselines, with up to $\mathit{49.5\times}$ faster speed or $\mathit{35\%}$ higher accuracy. 
	\method is fast, taking about {\it 2} seconds on a graph with 4.5 million edges. 
	}
\end{itemize}
Our code and data are publicly available\codeurl.

\begin{figure}[!t]
	\includegraphics[width=1\linewidth]{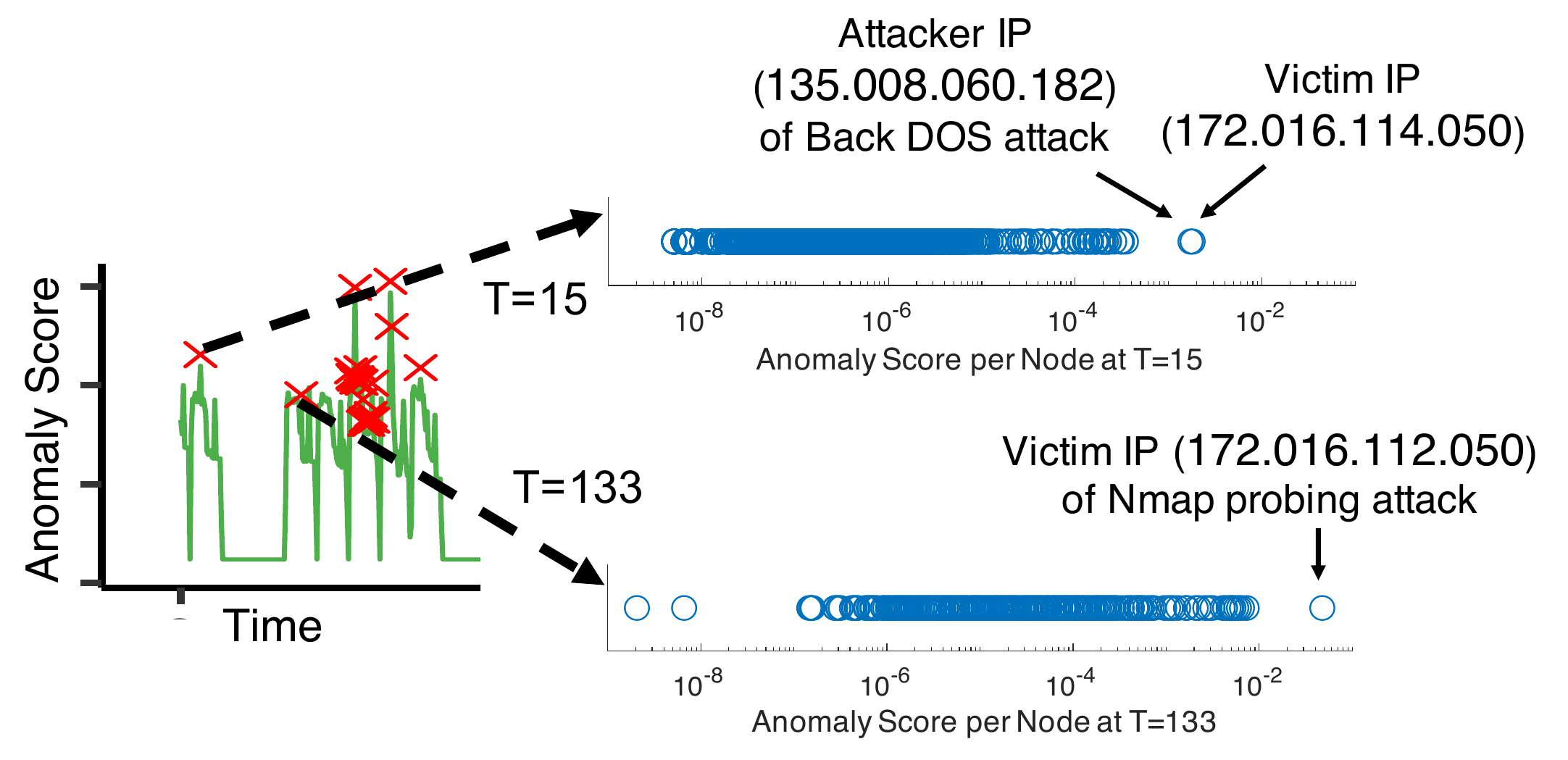}
	\caption
	{ 
		\label{fig:local}
		\underline{\smash{Attribution}}:
		\method localizes the culprits of anomalous events in the {\it DARPA} dataset:
		in a {\it Back} DOS attack, the attacker and victim IP have the top-2 largest scores;
		in an {\it Nmap} probing attack, the victim IP has the largest \method score.
	}
\end{figure}

%% file: 080acknowledgements.tex
This material is based upon work supported by the National Science Foundation under Grants No. CNS-1314632 and IIS-1408924. 
Any opinions, findings, and conclusions or recommendations expressed in this material are those of the author(s) and do not necessarily reflect the views of the National Science Foundation, or other funding parties.
The U.S. Government is authorized to reproduce and distribute reprints for Government purposes notwithstanding any copyright notation here on.

%% file: 070appendix.tex
\subsection{Experimental Setting}
\label{sec:app:exp}
All experiments are carried out on a 3 GHz Intel Core i5 iMac, 16 GB RAM, running OS X 10.13.6.
We implemented \method and SedanSpot in C++, and we used an open-sourced implementation of DenseAlert\footnote{https://github.com/kijungs/densealert}, provided by the authors of~\cite{shin2017densealert}.
To show the best trade-off between speed and accuracy, we set the sample size to 50 for SedanSpot and follow other parameter settings as suggested in the original paper~\cite{eswaran2018sedanspot}.
For \method, we set the damping factor $c$ to $0.5$, and stop iterations for computing node score vectors when $L1$ changes of node score vectors are less than $10^{-3}$.

\subsection{Dataset}
\label{sec:app:dataset}
\textit{\textbf{DARPA}}~\cite{lippmann1999results}
has 4.5M IP-IP communications between 9.4K source IP and 2.3K destination IP over 87.7K minutes.
Each communication is a directed edge ({\it srcIP, dstIP, timestamp, attack}) where the attack label indicates whether the communication is an attack or not.
We aggregate edges occurring in every hour, resulting in a stream of 1463 graphs.
We annotate a graph snapshot as anomalous if it contains at least 50 attack edges.
Then there are 288 ground truth anomalies ($23.8\%$ of total).
We use the first 256 graphs for initializing means and variances needed during normalization (as described in Section~\ref{sec:algorithm}).

\noindent\textit{\textbf{ENRON}}~\cite{shetty2004enron}
contains 50K emails from 151 employees over 3 years in the ENRON Corporation.
Each email is a directed edge ({\it sender, receiver, timestamp}).
We aggregate edges occurring in every day duration, resulting in a stream of 1139 graphs.
We use the first 256 graphs for initializing means and variances.

\noindent\textbf{RTM method}~\cite{akoglu2008rtm}
generates time-evolving graphs with repeated Kronecker products.
We use the publicly available code\footnote{\url{www.alexbeutel.com/l/www2013}}.
The generated graph is a directed graph with 1K nodes and 8.1K edges over 2.7K timestamps.
We use the first 300 timestamps for initializing means and variances.

\subsection{Proofs}
\label{sec:app:proofs}
We prove upper bounds on the 1st and 2nd derivatives of \prS and \prE, showing their effectiveness in detecting \structure and \edge. 
	\begin{proof} [Proof of Lemma~\ref{lemma:ps_one} (Upper bound of $\lVert\psone\rVert_{1}$)] $\\$
		\label{proof:ps_one}
		For brevity, $\psn \leftarrow \ps(t+\dt), \pso \leftarrow \ps(t)$.
		By Lemma~\ref{lemma:dynamic_ps}, $\lVert\psn - \pso\rVert_{1}$ is presented as follows:
		\small
		\begin{align*}
			\lVert\psn - \pso\rVert_{1} &= \lVert\sum_{k=0}^{\infty}c^k(\NATS+\DAS)^kc(\DAS\pso)\rVert_{1}\\
			&\le c\sum_{k=0}^{\infty}\lVert c^k(\NATS+\DAS)^k\rVert_{1}\lVert\DAS\pso\rVert_{1}\\
			&\le \frac{c}{1-c}\lVert\DAS\pso\rVert_{1} \le \frac{c}{1-c}\lVert\DAS\rVert_{1}\\
			\lVert\psone\rVert_{1} &= \lVert\frac{\psn - \pso}{\dt}\rVert_{1} \le \frac{c}{1-c}\lVert\frac{\DAS}{\dt}\rVert_{1}
		\end{align*}
		\normalsize
		Note that $\lVert(\NATS+\DAS)^k\rVert_{1} = \lVert\pso\rVert_{1} = 1$ since $(\NATS+\DAS)$ is a column-normalized stochastic matrix, and $\pso$ is a PageRank vector.
	\end{proof}

	\begin{proof}[Proof of Lemma~\ref{lemma:ps_two} (Upper bound of $\lVert\pstwo\rVert_{1}$)] $\\$
		\label{proof:ps_two}
		For brevity, $\p_0 \leftarrow \ps(t-\Delta t), \p_1 \leftarrow \ps(t), \p_2 \leftarrow \ps(t+\Delta t),  \Delta\p^o \leftarrow \p_1-\p_0, \Delta\p^n \leftarrow \p_2-\p_1, \A \leftarrow \NATS, \DA_1 \leftarrow \DASO$ and $\DA_2 \leftarrow \DASN$.
		In addition, we omit $c$ by substituting $\A \leftarrow c\A$ and $\DA \leftarrow c\DA$ during this proof.
		By Lemma~\ref{lemma:dynamic_ps}, $\Delta\p^n$ is:
		\small
		\begin{align*}
		\Delta\p^n &= \sum_{k=0}^{\infty}(\A + \DA_1 + \DA_2)^k(\DA_2\p_1)
		\end{align*}
		\normalsize
		$\Delta\p^n$ can be viewed as an updated \prS with the original adjacency matrix $Y_1 = (\A + \DA_1)$, the update $\DA_2$, and the starting vector $(\DA_2\p_1)$ from an original vector $\p_{temp}=\sum_{k=0}^{\infty}Y_1^k(\DA_2\p_1)$.
		Then, by Lemma~\ref{lemma:dynamic_ps}, $\Delta\p^n$ is presented as follows:
		\small
		\begin{align*}
		\Delta\p^n &= \p_{temp} + \sum_{k=0}^{\infty}(Y_1 + \DA_2)^k\DA_2\p_{temp}\\
		&= \sum_{k=0}^{\infty}Y_1^k(\DA_2\p_1)+ \sum_{k=0}^{\infty}(Y_1 + \DA_2)^k\DA_2\sum_{i=0}^{\infty}Y_1^i(\DA_2\p_1)
		\end{align*}
		\normalsize
		Then $\Delta\p^n - \Delta\p^o$ becomes as follows: 
		\small
		\begin{align*}
		\Delta\p^o &= \sum_{k=0}^{\infty}(\A + \DA_1)^k(\DA_1\p_0) = \sum_{k=0}^{\infty}Y_1^k(\DA_1\p_0)\\
		\Delta\p^n - \Delta\p^o = &\sum_{k=0}^{\infty}Y_1^k(\DA_2-\DA_1)\p_1 + \sum_{k=0}^{\infty}Y_1^k\DA_1(\p_1-\p_0) \\
		&+ \sum_{k=0}^{\infty}(Y_1 + \DA_2)^k\DA_2\sum_{i=0}^{\infty}Y_1^i(\DA_2\p_1)
		\end{align*}
		\normalsize
		Since $\p_1-\p_0 = \Delta\p^o = \sum_{k=0}^{\infty}Y_1^k(\DA_1\p_0)$, $\\$
		the second term $\sum_{k=0}^{\infty}Y_1^k\DA_1(\p_1-\p_0)$ in the equation above is:
		\small
		\begin{align*}
		\sum_{k=0}^{\infty}Y_1^k\DA_1(\p_1-\p_0) = \sum_{k=0}^{\infty}Y_1^k\DA_1 \sum_{i=0}^{\infty}Y_1^i(\DA_1\p_0)
		\end{align*}
		\normalsize
		Note that$\lVert\p_0\rVert_{1} = \lVert\p_1\rVert_{1} = 1$ since $\p_0$ and $\p_1$ are PageRank vectors.
		Recovering $c$ from $\A$ and $\DA$, $\lVert\sum_{k=0}^{\infty}Y_1^k\rVert_{1}$ and $\lVert\sum_{k=0}^{\infty}(Y_1+\DA_2)^k\rVert_{1}$ becomes as follows:
		\small
		\begin{align*}
		\lVert\sum_{k=0}^{\infty}Y_1^k\rVert_{1} &= \lVert \sum_{k=0}^{\infty}c^k(\A + \DA_1)^k\rVert_{1} = \frac{1}{1-c} \\
		\lVert\sum_{k=0}^{\infty}(Y_1+\DA_2)^k\rVert_{1} &= \lVert \sum_{k=0}^{\infty}c^k(\A + \DA_1 + \DA_2)^k\rVert_{1} = \frac{1}{1-c} 
		\end{align*}
		\normalsize
		Note that $\A+\DA_1$ and $\A+\DA_1+\DA_2$ are column-normalized stochastic matrices.
		Then $\lVert\Delta\p^n - \Delta\p^o\rVert_{1}$ is bounded as follow:
		\small
		\begin{align*}
		\lVert\Delta\p^n - \Delta\p^o\rVert_{1} \le \frac{c}{1-c}\lVert\DA_2-\DA_1 \rVert_{1} + (\frac{c}{1-c})^2(\lVert\DA_1\rVert_{1}^2 + \lVert\DA_2\rVert_{1}^2)
		\end{align*}
		\normalsize
		Then, recovering $c$ from all terms, $\lVert\pstwo\rVert_{1}$ is bounded as follows:
		\small
		\begin{align*}
		&\lVert\pstwo\rVert_{1} = \frac{\lVert\Delta\p^n - \Delta\p^o\rVert_{1}}{\dt^2} \\
		&\le \frac{\frac{c}{1-c}\lVert\DA_2-\DA_1 \rVert_{1} + (\frac{c}{1-c})^2(\lVert\DA_1\rVert_{1}^2 + \lVert\DA_2\rVert_{1}^2)}{\dt^2}
		\end{align*}
		\normalsize
	\end{proof}

	\begin{proof}[Proof of Theorem~\ref{theorem:up_d1d2_s}] (Upper bounds of $\lVert\psone\rVert_{1}$ and $\lVert\pstwo\rVert_{1}$ with \structure)
		\label{proof:up_d1d2_s}	
		Use Lemma~\ref{lemma:ps_structure} and~\ref{lemma:ps_one}.
	\end{proof}

	\begin{proof}[Proof of Lemma~\ref{lemma:pe_one} (Upper bounds of $\lVert\peone\rVert_{1}$)] $\\$
		\label{proof:pe_one}
		For brevity, denote $\peo \leftarrow \pe(t)$ and $\pen \leftarrow \pe(t+\Delta t)$. 
		By Lemma~\ref{lemma:dynamic_pe}, $\lVert\pen - \peo\rVert_{1}$ is presented as follows:
		\small
		\begin{align*}
		\pen - \peo =~& \sum_{k=0}^{\infty}c^k(\NATE+\DAE)^kc\DAE\peo\\ 
		&+(1-c)\sum_{k=0}^{\infty}c^k(\NATE+\DAE)^k\bed\\
		\lVert\pen - \peo\rVert_{1} \le~&\frac{c}{1-c}\lVert\DAE\rVert_{1}+\lVert\bed\rVert_{1}\\
		\lVert\peone\rVert_{1} =~&\frac{\lVert\pen - \peo\rVert_{1}}{\dt} \le \frac{1}{\dt}(\frac{c}{1-c}\lVert\DAE\rVert_{1}+\lVert\bed\rVert_{1})
		\end{align*}
		\normalsize
		$\lVert(\NATE+\DAE)^k\rVert_{1} = \lVert\peo\rVert_{1} = 1$ since $\NATE+\DAE$ is a column-normalized stochastic matrix and $\peo$ is a PageRank vector.
	\end{proof}

	\begin{proof}[Proof of Lemma~\ref{lemma:pe_two} (Upper bound of $\lVert\petwo\rVert_{1}$)] $\\$
		\label{proof:pe_two}
		For brevity, denote $\p_0 \leftarrow \pe(t-\Delta t), \p_1 \leftarrow \pe(t), \p_2 \leftarrow \pe(t+\Delta t), \Delta\p^o \leftarrow \p_1-\p_0, \Delta\p^n \leftarrow \p_2-\p_1, \A \leftarrow \NATE, \DA_1 \leftarrow \DAEO, \DA_2 \leftarrow \DAEN, \bd_1 \leftarrow \bedo$ and $\bd_2 \leftarrow \bedn$.
		In addition, we omit the $c$ term by substituting $\A \leftarrow c\A, \DA \leftarrow c\DA$ and $\bd \leftarrow (1-c)\bd$ during this proof.
		By Lemma~\ref{lemma:dynamic_pe}, $\Delta\p^o$ and $\Delta\p^n$ are presented as follows:
		\small
		\begin{align*}
		\Delta\p^o &= \sum_{k=0}^{\infty}(\A+\DA_1)^{k}\DA_1\p_0 + \sum_{k=0}^{\infty}(\A+\DA_1)^{k}\bd_1 \\
		\Delta\p^n &= \sum_{k=0}^{\infty}(\A+\DA_1+\DA_2)^{k}\DA_2\p_1 + \sum_{k=0}^{\infty}(\A+\DA_1 + \DA_2)^{k}\bd_2
		\end{align*}
		\normalsize
		Substracting the first term of $\Delta\p^o$ from the first term of $\Delta\p^n$ is equal to $\pstwo$ as shown in Lemma~\ref{lemma:ps_two}.
		Then $\Delta\p^n - \Delta\p^o$ is:
		\small
		\begin{align*}
		\Delta\p^n - \Delta\p^o = \pstwo +  \sum_{k=0}^{\infty}(\A+\DA_1 + \DA_2)^{k}\bd_2 - \sum_{k=0}^{\infty}(\A+\DA_1)^{k}\bd_1 
		\end{align*}
		\normalsize
		\vfill\eject 
		By substituting $\A +\DA_1$ with $Y_2$, the last two terms in the above equation are presented as follows:
		\small
		\begin{align*}
		&\sum_{k=0}^{\infty}(Y_2+\DA_2)^k\bd_2-\sum_{k=0}^{\infty}Y_2^k\bd_1\\
		&=\sum_{k=0}^{\infty}Y_2^k\bd_2 + \sum_{k=0}^{\infty}(Y_2+\DA_2)^i\DA_2\sum_{i=0}^{\infty}Y_2^k\bd_2 - \sum_{k=0}^{\infty}Y_2^k\bd_1\\
		&=\sum_{k=0}^{\infty}Y_2^k(\bd_2-\bd_1) + \sum_{k=0}^{\infty}(Y_2+\DA_2)^i\DA_2\sum_{i=0}^{\infty}Y_2^k\bd_2
		\end{align*}
		\normalsize
		In the first equation, we treat $\sum_{k=0}^{\infty}(Y_2+\DA_2)^k\bd_2$ as an updated PageRank with the update $\DA_2$ from an original PageRank $\sum_{k=0}^{\infty}Y_2^k\bd_2$, then apply Lemma~\ref{equ:update_ps}.
		Note that both $\lVert\sum_{k=0}^{\infty}Y_2^k\rVert_{1}$ and $\lVert\sum_{k=0}^{\infty}(Y_2+\DA_2)^k\rVert_{1}$ have value $\frac{1}{1-c}$ since the original expressions with $c$ terms are as follows:
		\small
		\begin{align*}
		\sum_{k=0}^{\infty}Y_2^k&= \sum_{k=0}^{\infty}c^k(\NATE+\DAEO)^k\\
		\sum_{k=0}^{\infty}(Y_2+\DA_2)^k&= \sum_{k=0}^{\infty}c^k(\NATE+\DAEO+\DAEN)^k
		\end{align*}
		\normalsize
		$(\NATE+\DAEO)$ and $(\NATE+\DAEO+\DAEN)$ are column-normalized stochastic matrices.
		Then, recovering $c$ from all terms, $\lVert\petwo\rVert_{1}$ is bounded as follow:
		\small
		\begin{align*}
		\frac{\lVert\Delta\p^n - \Delta\p^o\rVert_{1}}{\dt^2} &\le \lVert\pstwo\rVert_{max} \\
		&+ \frac{1}{\dt^2}(\lVert\bd_2-\bd_1 \rVert_{1} + \frac{c}{1-c}\lVert\DA_2\rVert_{1}\lVert\bd_2\rVert_{1})
		\end{align*}
		\normalsize
	\end{proof}

	\begin{proof}[Proof of Theorem~\ref{theorem:up_d1d2_e}] (Upper bounds of $\lVert\peone\rVert_{1}$ and $\lVert\petwo\rVert_{1}$ with \edge)
		\label{proof:up_d1d2_e}
		Use Lemma~\ref{lemma:pe_edge},~\ref{lemma:ps_two} and~\ref{lemma:pe_two}.
	\end{proof}